\documentclass[10pt,twocolumn,fleqn]{article}
\usepackage{geometry}
\usepackage[sc]{mathpazo}
\usepackage[mathscr]{eucal}
\usepackage{enumerate}
\usepackage{booktabs}
\usepackage{tabularx}
\usepackage{hyperref}
\usepackage{multirow}
\usepackage{cite}
\usepackage{mathptmx} 
\usepackage{amsmath,amsfonts,amssymb,amsthm}
\usepackage{caption}
\usepackage{subfigure}
\usepackage{tgpagella}
\usepackage{graphicx}
\usepackage{xcolor}
\usepackage[textwidth=3cm,textsize=tiny,shadow]{todonotes}
\newcommand\cM{{\mathcal M}}

\newcommand\scM{{\mathscr M}}

\newcommand\mvector{\boldsymbol}

\newcommand\vf{\mvector{f}}

\newcommand\vp{\mvector{p}}

\newcommand\vq{\mvector{q}}

\newcommand\vu{\mvector{u}}
\newcommand\vv{\mvector{v}}

\newcommand\vx{\mvector{x}}
\newcommand\vy{\mvector{y}}

\newcommand\vP{\mvector{P}}

\newcommand\vX{\mvector{X}}

\newcommand\vGamma{\mvector{\Gamma}}

\newcommand\field{\mathbb}

\newcommand\R{\field{R}}

\newcommand\C{\field{C}}
\newcommand\Z{\field{Z}}

\newcommand\Q{\field{Q}}

\newcommand\diag{\operatorname{diag}}

\newcommand\id{\operatorname{\mathrm{Id}}}

\newcommand\rmd{\mathrm{d}}

\newcommand\rmi{\mathrm{i}\mspace{1mu}}
\newcommand\rme{\mathrm{e}}

\newcommand\pder[2]{\dfrac{\partial #1 }{\partial #2}}

\newcommand\abs[1]{\lvert #1 \rvert}

\usepackage{mathtools}
\theoremstyle{plain}
\newtheorem{theorem}{Theorem}
\newtheorem{lemma}[theorem]{Lemma}
\newtheorem{proposition}[theorem]{Proposition}

\newtheoremstyle{note}{\topsep}{\topsep}{\slshape}{}{\scshape}{}{ }{}
\theoremstyle{note}

\theoremstyle{remark}
\newtheorem{remark}[theorem]{Remark}
\numberwithin{equation}{section}
\numberwithin{theorem}{section}
\usepackage{enumitem}

\setlist[description]{font=\normalfont}
\usepackage{scalerel,stackengine}
\stackMath
\newcommand\reallywidehat[1]{%
\savestack{\tmpbox}{\stretchto{%
  \scaleto{%
    \scalerel*[\widthof{\ensuremath{#1}}]{\kern.1pt\mathchar"0362\kern.1pt}%
    {\rule{0ex}{\textheight}}
  }{\textheight}%
}{2.4ex}}%
\stackon[-6.9pt]{#1}{\tmpbox}%
}

\title{Integrability of non-homogeneous Hamiltonian systems with gyroscopic coupling}
\author{
Wojciech Szumi{\'n}ski$^{1}$ and Andrzej J. Maciejewski$^{2}$ \\
Institute of Physics$^{1}$ and  Janusz Gil Institute of Astronomy$^{2}$ \\
 University of Zielona G\'ora, Licealna 9, PL-65--417,  Zielona G\'ora, Poland
\\ e-mail: w.szuminski@if.uz.zgora.pl
}

\begin{document}
\maketitle
\begin{abstract}
We study the integrability of a two-dimensional Hamiltonian system with a
gyroscopic term and a non-homogeneous potential composed of two homogeneous
components of different degrees. The model describes the motion of a particle in
a plane under the combined influence of a central (Kepler-type) potential, a
uniform magnetic field, and a superposition of homogeneous forces. By
combining the Levi--Civita regularization with the so-called coupling constant
metamorphosis transformation, and employing differential Galois theory, we
derive analytical necessary conditions for integrability in the Liouville sense.  They put restrictions on the degrees of homogeneity of the potential terms and their values in particular points.  The obtained results
encompass and generalize several classical galactic and astrophysical models,
including the generalized Hill model, the H\'enon--Heiles and
Armbruster--Guckenheimer--Kim systems, providing a unified framework for
studying non-homogeneous  Hamiltonians. We demonstrate the
effectiveness of the derived integrability obstructions by proving the
non-integrability of these models in the presence of a uniform rotational field.
The numerical analysis via the Poincar\'e cross-sections further confirms the
analytical results, illustrating the transition from regular to chaotic dynamics
as the rotational and non-homogeneous terms are introduced. Moreover, we show
that, without the Kepler-type term, a generalized non-homogeneous extension of
the exceptional potential remains integrable. The explicit forms of the first integrals are given.

\paragraph{Declaration} The article has been published by \textit{Nonlinear Dynamics} in~\cite{SzuminskiMaciejewski2026}, and the final version is available at: \textbf{\href{https://doi.org/10.1007/s11071-026-12309-x}{https://doi.org/10.1007/s11071-026-12309-x} }
\end{abstract}

\section{Motivation and description of the system}
\label{sec:intro}
The problem of determining whether a given dynamical system is integrable or not
is one of the central topics in the theory of differential equations and
Hamiltonian systems. While the integrability of homogeneous potentials has been
extensively studied and is now relatively well
understood~\cite{Morales:00::,Maciejewski:04::, Szuminski:15::} much less is
known about systems that include additional
non-homogeneous~\cite{Yoshida:88::,Mondejar:99::}, or gyroscopic
terms~\cite{2023NonDy.111..275P}, which naturally arise in many physical
contexts.  Such systems often serve as simplified models of galactic
dynamics~\cite{Caranicolas_1989, Caranicolas_1990a},  
the motion of charged particles in magnetic fields~\cite{Kubu_2024},  
or nonlinear oscillations in rotating frames~\cite{Innanen_1985,Lanchares_2021},  
where the interplay between symmetry and rotation gives rise to complex dynamical structures.  
For a more detailed discussion of these phenomena, see the monograph~\cite{Contopoulos_2002}.

In particular, the inclusion of a gyroscopic (Coriolis-like) term fundamentally alters 
the dynamics, leading to new mechanisms of resonance and chaos~\cite{Hill_1905,Vallado_2007, Giacaglia_2019}. 
At the same time, adding a non-homogeneous contribution to the potential 
breaks the scaling symmetry that usually facilitates analytical integration~\cite{Ito_1987}. 
The combined effect of these two perturbations --- rotation and non-homogeneity --- 
poses a challenging question regarding the persistence or loss of integrability.  

In this paper, we address this problem by analysing a class of two--degree--of--freedom Hamiltonian system governed by Hamiltonian of the form
\begin{align}
\label{eq:Hamek}
	H_\mu=\frac{1}{2}\left(p_1^2+p_2^2\right)
	+\omega(q_2p_1-q_1p_2)
	-\frac{\mu}{r}
	+V(q_1,q_2),
\end{align}
where the potential $V$ is expressed as
\begin{equation}
\label{eq:VV}
V(q_1,q_2)=V_k(q_1,q_2)+V_m(q_1,q_2).
\end{equation}
Hamilton equations determined by the Hamilton function~\eqref{eq:Hamek} have the form
\begin{equation}
\begin{aligned}
\dot q_1 &= p_1 + \omega q_2, 
&\qquad 
\dot p_1 &= \omega p_2 - \frac{\mu q_1}{r^{3}} - \frac{\partial V}{\partial q_1},
\\[1mm]
\dot q_2 &= p_2 - \omega q_1, 
&\qquad
\dot p_2 &= -\omega p_1 - \frac{\mu q_2}{r^{3}} - \frac{\partial V}{\partial q_2}.
\end{aligned}
\label{eq:Ham-eq}
\end{equation}
Here, \(r=\sqrt{q_1^2+q_2^2}\) denotes the distance from the origin.  
The functions \(V_i(q_1,q_2)\) are assumed to be homogeneous of rational degrees \(i\in\mathbb{Q}\setminus\{-1,0\}\), with \(k\neq m\).  
The parameter \(\mu\) measures the strength of the central (Kepler-type) potential, while \(\omega\) represents the frequency associated with the gyroscopic term.  
Throughout the main part  of the paper, we assume $\mu\omega\neq 0$, and under this assumption, we will derive necessary conditions for the Liouville integrability of the system.
The degenerate case $\mu=0$, is treated separately in the final part of the paper.

Hamiltonian~\eqref{eq:Hamek} describes the motion of a particle in a plane under the combined influence of a central potential \(-\mu/r\), two homogeneous potentials \(V_k\) and \(V_m\), and the gyroscopic term \(\omega(q_2p_1-q_1p_2)\) corresponding to motion in a uniformly rotating reference frame. 
The gyroscopic term introduces coupling between the canonical coordinates and momenta, generating effective magnetic-like forces~\cite{Royer_2011,Brandao_2015} that significantly alter the dynamics, consequently, the integrability properties of the system~\cite{Combot_2022,2023NonDy.111..275P}.

Physically, Hamiltonian~\eqref{eq:Hamek} represents a broad class of two-dimensional galactic and astrophysical models describing the motion of a test particle in a rotating frame~\cite{Lanchares_2021,Inarrea_2015}, see also the book~\cite{Binney_Tremaine_2008}. 
The term \(-\mu/r\) accounts for the gravitational attraction of a central mass, whereas the homogeneous components \(V_k\) and \(V_m\) model deviations from spherical symmetry in the galactic or stellar potential. 
The presence of the rotational term \((\omega\neq 0)\) produces characteristic effects such as resonance capture, bifurcations of periodic orbits, and the formation of stochastic layers that separate regions of regular motion\cite{Elmandouh2016,Salas_2022,Lacomba_2012}.

Moreover, Hamiltonian~\eqref{eq:Hamek} admits an equivalent electrodynamic interpretation. 
It can be written in the compact vector form
\begin{equation}
\label{eq:HH1}
H_\text{EM}
=\frac{1}{2}\bigl(\boldsymbol{p}-\boldsymbol{A}(q_1,q_2)\bigr)^2
-\frac{\mu}{r}
+V(q_1,q_2),
\end{equation}
where \(\boldsymbol{A}(q_1,q_2)\) is the vector potential of a uniform magnetic field 
\(\boldsymbol{B}=\nabla\times\boldsymbol{A}=2\omega\,\hat{\boldsymbol{z}}\). 
For the symmetric gauge \(\boldsymbol{A}=(-\omega q_2,\,\omega q_1,\,0)\), 
Hamiltonian~\eqref{eq:HH1} reduces exactly to~\eqref{eq:Hamek} with 
\[
V(q_1,q_2)=\frac{1}{2}\omega^2\left(q_1^2+q_2^2\right)+V_m(q_1,q_2).
\]
This form reveals that the rotation term introduces an additional quadratic contribution to the potential, which breaks its original homogeneity.  
Hence, the study of the non-homogeneous potential~\eqref{eq:VV} is crucial: 
it captures how the combined effects of the magnetic (gyroscopic) field and the nonlinear term \(V_m(q_1,q_2)\) modify the integrability and dynamical structure of the system.  
In this sense, the same mathematical framework describes the planar motion of a charged particle in a constant perpendicular magnetic field, subject to a central Coulomb potential and an additional external potential \(V(q_1,q_2)\).

Such Hamiltonians generalize several classical models in celestial mechanics and galactic dynamics, including the H\'enon--Heiles system~\cite{Henon:64::}, the Armbruster--Guckenheimer--Kim galactic model~\cite{Armbruster_1989}, and the planar Hill or restricted three-body problems~\cite{Hill_1905,Beletsky_2001,Chauvineau_1990}. 
In all these cases, the interplay between the central potential, the anisotropic perturbations, and the gyroscopic term gives rise to rich nonlinear behaviour, ranging from integrable regimes to fully developed chaos. 
Understanding how the gyroscopic term modifies the integrability of these systems is therefore essential for explaining the stability of stellar orbits, the morphology of rotating galaxies, and the onset of chaotic transport in gravitational and electromagnetic systems.

The principal goal of this work is to determine how the presence of the gyroscopic term alters the integrability conditions of generalized galactic-type Hamiltonians. 
In particular, we aim to identify the forms of the potential \(V(q_1,q_2)\) and parameter combinations \((k,m)\) for which system~\eqref{eq:Hamek} may still admit additional meromorphic first integrals. 
This classification bridges the gap between classical non-rotating integrable models and their rotating analogues, providing new insights into the transition from integrable to chaotic dynamics in low-dimensional Hamiltonian systems with rotational symmetry.

The main result of this paper is formulated in the following theorem, which provides necessary conditions for the Liouville integrability of system~\eqref{eq:Hamek}.  
To express it in a compact form, we introduce two auxiliary rational parameters:
\begin{align}
\label{eq:n,l}
n := \frac{2 - k}{m - k}, 
\qquad 
l := \frac{2 + 3k}{m - k},
\end{align}
and define the \textit{integrability coefficients} as
\begin{equation}
\label{eq:coeffs}
\lambda_k := V_k(1, \rmi), 
\qquad 
\lambda_m := V_m(1, \rmi).
\end{equation}
We note that, in general, $\lambda_k, \lambda_m \in \C$.

\begin{theorem}[Main]
\label{th:main_theorem}
Assume that $\mu\,\omega\neq 0$ and $\lambda_k\,\lambda_m\neq0$.  
If system~\eqref{eq:Hamek} is Liouville integrable with meromorphic first integrals, then:
\begin{enumerate}
  \item \label{it:necs1} $l\geq-1$ is an odd, or $l<-1$ is an even integer, or 
  \item \label{it:necs2} $(n+l)(n+l+2)\neq 0$ and either 
  \begin{enumerate}
    \item $n>0$ is an even, or $n<0$ is an odd integer; or
    \item \label{it:necs3} $n+l$ is an even integer, except the case when $l\geq 0$ and $n\leq 0$ are both even integers.
  \end{enumerate}
\end{enumerate}
\end{theorem}

Theorem~\ref{th:main_theorem} establishes the principal integrability obstructions for the general case in which both coefficients, $\lambda_k$ and~$\lambda_m$, are nonzero.
As will be demonstrated later, these obstructions are particularly strong and highly effective in practical applications. It is remarkable that they depend solely on the degrees of homogeneity \(k\) and \(m\) of the potential components, and values of the potential at the specific complex point~\((1, \rmi)\).

However, certain degenerate configurations occur when one of these coefficients vanishes, leading to qualitatively different dynamical behaviours that require a separate analysis.

In particular, the situation when either $\lambda_k$ or~$\lambda_m$ is nonzero for $n = 0$ (corresponding to $k = 2$) demands a more detailed investigation.
This is motivated by the existence of a special class of non-homogeneous potentials in physics and astronomy whose lower-degree term is quadratic ($k = 2$).
Classical examples include the mentioned Hill problem, H\'enon-Heiles, and the generalized Armbruster-Guckenheimer-Kim galactic potentials, which share this structure and can be regarded as non-homogeneous potentials with a quadratic leading part.
Therefore, this exceptional case is treated separately, and the following two theorems provide the corresponding obstructions to integrability.

\begin{theorem}
\label{th:main_theorem_k=2_m}
Assume that $\mu\, \omega \neq 0$, $\lambda_k~=~0$, and $\lambda_m~\neq~0$.  
If~$k~\in~\mathbb{Q}$ and $\abs{m} > 2$, then the system~\eqref{eq:Hamek} is not Liouville integrable with meromorphic first integrals.
\end{theorem}

\begin{theorem}
\label{th:main_theorem_k=2}
Assume that $\mu\,\omega \neq 0$, $\lambda_k \neq 0$, and $\lambda_m~=~0$.  
If~$k~=~2$ and $m \in \mathbb{Q}$, then the system~\eqref{eq:Hamek} does not admit any meromorphic first integral functionally independent of the Hamiltonian.
\end{theorem}
\vspace{1em}

In the above theorems by meromorphic first integrals, we understand complex meromorphic functions of variables $(q_1,q_2,p_1,p_2, r)$.

The proofs of Theorems~\ref{th:main_theorem}–\ref{th:main_theorem_k=2} are based on the Morales--Ramis theory~\cite{MR2647643,Morales:99::}, 
which provides one of the most effective analytical tools for studying the (non-)integrability of Hamiltonian systems.  
This approach links the classical idea of linearisation around a particular solution with the modern framework of differential Galois theory, 
establishing a direct correspondence between the algebraic structure of the variational equations 
and the dynamical properties of the original nonlinear system~\cite{Put:03::,Audin:08::}.

The central result in this context was established by Morales and Ramis~\cite{Morales:99::}, 
who proved that Liouville integrability imposes strong algebraic constraints on the Galois group of the variational equations.

\begin{theorem}[Morales--Ramis, 1999]
\label{th:MR}
Let a complex Hamiltonian system with $n$ degrees of freedom be Liouville integrable 
with meromorphic first integrals in a neighbourhood of a non-equilibrium phase curve~$\vGamma$.  
Then, the identity component of the differential Galois group of the variational equations along~$\vGamma$ is Abelian.
\end{theorem}

This theorem provides a necessary condition for integrability and serves as the basis of most modern non-integrability proofs.  
In practice, its application follows a relatively standard sequence of steps:  
first, one identifies a particular solution~$\vGamma(t)$  of  equations of motion~\eqref{eq:Ham-eq} generated by Hamiltonian~\eqref{eq:Hamek}.  
Next, the equations of motion are linearised along this trajectory to obtain the variational equations.  
Finally, one analyzes the differential Galois group of these equations.  
If the identity component of this group is shown to be non-Abelian, 
then, by Theorem~\ref{th:MR}, the original system cannot be Liouville integrable.

To prove Theorems~\ref{th:main_theorem}–\ref{th:main_theorem_k=2}, 
we therefore construct an appropriate particular solution of the system~\eqref{eq:Ham-eq} 
and demonstrate that the associated variational equations possess a non-Abelian Galois group.  
The explicit construction of this solution and the derivation of the corresponding variational equations 
are discussed in the next section.

It should be emphasised that, in general, there is no systematic or algorithmic method 
for selecting a trajectory~$\vGamma(t)$ suitable for the application of the Morales--Ramis theory.  
Even when such a trajectory is known, the subsequent analysis of the differential Galois group 
is often highly non-trivial.  
In most cases, the variational equations do not decouple into lower-dimensional subsystems, 
and their Galois groups must be studied case by case using a combination of algebraic and analytic arguments.

A comprehensive exposition of the theory and its applications can be found in 
the works of Morales and Ramis~\cite{Morales:99::,MR2647643} and 
the monograph by Audin~\cite{Audin:08::}.  
For an accessible introduction with worked examples, see also~\cite{Maciejewski:03::a}.  
The Morales--Ramis framework has become a standard tool for detecting non-integrability, 
and it has been successfully applied to a wide variety of problems — 
ranging from the classical three-body problem~\cite{Boucher:03::,Maciejewski:10::,Maciejewski:11::}, $n$-body problem~\cite{Maciejewski_2025} to modern models in galactic dynamics~\cite{Acosta_2018}.  Numerous other applications of this theory in recent years can be found in~\cite{Yagasaki:18::,Huang:18::,Combot:18::,Mnasri:18::,Shibayama:18::,Elmandouh:18::,Szuminski:18a::,Maciejewski:18::,Maciejewski:17::,Szuminski:16::,Szuminski:24::, Szuminski_2025_JSV}.

Following the assumptions of Theorem~\ref{th:MR}, 
we consider the complexified version of our system, 
$(q_1,q_2,p_1,p_2)\in\C^4$, 
with the potential $V(q_1,q_2)$ assumed to be algebraic over $\C(q_1,q_2)$.  
Although the Hamiltonian~\eqref{eq:Hamek} is not strictly meromorphic 
because of the term $-\mu/r=-\mu/\sqrt{q_1^2+q_2^2}$, 
it has been shown in~\cite{Combot:13::,Maciejewski:16::} 
that the Morales--Ramis theory can still be consistently applied to such systems, 
provided that the singularities of the potential are handled within the framework of algebraic differential equations. However, we avoid difficulties just by considering integrability in terms of complex meromorphic functions of variables $(q_1,q_2,p_1,p_2,r)$. Moreover, the application of the Levi-Civita transformation reduced the problem to studying the integrability of systems with a rational  Hamiltonian function.  

The rest of the paper is organized as follows. In Sec.~\ref{sec:particular_solutions} we construct explicit particular solutions of the Hamiltonian $H_\mu$,   using the Levi--Civita regularization combined with the coupling--constant metamorphosis. We then derive the variational equations restricted to an invariant plane and rewrite them as two second-order reduced differential equations: a homogeneous Gauss hypergeometric equation and a non-homogeneous equation sharing the same homogeneous part. In Secs.~\ref{sec:proof_main_theorem}--\ref{sec:proof_lemma_m} we analyse the integrability of the reduced variational equations via differential Galois theory and monodromy of the Gauss hypergeometric equation and its degenerate cases, identifying all admissible configurations for which the identity component of the Galois group is Abelian. These sections contain the proofs of Theorems~\ref{th:main_theorem}--\ref{th:main_theorem_k=2}. Sec.~\ref{sec:app} applies the obtained integrability obstructions to several classical models, including the generalized Hill, the H\'enon--Heiles, and the Armbruster--Guckenheimer--Kim systems, and illustrates the analytical results with representative Poincar\'e cross-sections. Sec.~\ref{sec:exc} departs from the obstruction-based analysis and examines the special case $\mu=0$, where Hamiltonian~\eqref{eq:Hamek} reduces to the rotating Hamiltonian $H_0$. In this regime the regularisation is no longer equivalent to the original dynamics, so Theorems~\ref{th:main_theorem}--\ref{th:main_theorem_k=2} do not apply, and the Morales--Ramis method cannot be used due to the lack of an appropriate particular solution. Nevertheless, we show that the rotating Hamiltonian $H_0$ with a non-homogeneous exceptional potential becomes super-integrable; Sec.~\ref{sec:exc} establishes this by constructing two additional independent first integrals. Sec.~\ref{sec:conclusions} provides concluding remarks and outlines further perspectives. The paper ends with two appendices: Appendix~\ref{app:criterion} formulates the analytic criterion determining when the identity component of the differential Galois group of the reduced variational system is Abelian, while Appendix~\ref{app:monodromy} presents the monodromy analysis of the relevant Gauss hypergeometric equation, including local monodromy matrices, connection formulas, and the logarithmic cases required in several proofs.

\section{Particular solutions and variational equations}
\label{sec:particular_solutions}

As noted above, the Morales--Ramis approach requires the existence of a
non-equilibrium particular solution of the equations of motion.
Since no general method for constructing such solutions is available, we employ
a sequence of canonical transformations adapted to the structure of the
Hamiltonian system under consideration.

We emphasize that these transformations are not introduced solely to regularize
the Kepler-type term $-\mu/r$, as the Morales--Ramis theory is also applicable to
Hamiltonian systems with algebraic potentials.
Their primary purpose is to identify invariant manifolds and to construct an
explicit particular solution of the system~\eqref{eq:Ham-eq} along which the
Morales--Ramis integrability analysis can be effectively carried out.

With this aim, we first apply the Levi--Civita transformation, which yields a
Hamiltonian form suitable for the subsequent application of the coupling
constant metamorphosis. Namely
\begin{equation}
\begin{aligned}
\label{eq:zamianka}
	q_1&=u_1^2-u_2^2, && p_1=\frac{u_1v_1-u_2v_2}{2(u_1^2+u_2^2)},\\
	 q_2&=2u_1u_2,&&
	p_2=\frac{u_1v_2+u_2v_1}{2(u_1^2+u_2^2)}.
\end{aligned}
\end{equation}
Hamiltonian~\eqref{eq:Hamek} in these new coordinates now reads
\begin{equation}
  \label{eq:hamek2}
\begin{split}
\widetilde{K}_\mu=&\frac{v_1^2+v_2^2}{8(u_1^2+u_2^2)}+\frac{\omega}{2}(u_2v_1-u_1v_2) -\frac{\mu}{u_1^2+u_2^2}+ 
	U(u_1,u_2),
	\end{split}
\end{equation}
where $U(u_1,u_2)=V(u_1^2-u_2^2,2u_1u_2)$.

The corresponding equations of motion are as follows
\begin{equation*}
    \begin{split}
        \label{eq:rhs}
     \dot u_1&=\frac{\partial \widetilde{K}_\mu}{\partial v_1}=
     \frac{v_1}{4(u_1^2+u_2^2)}+\frac{\omega}{2} u_2,\\
        \dot u_2&=\frac{\partial \widetilde{K}_\mu}{\partial v_2}=
        \frac{v_2}{4(u_1^2+u_2^2)}-\frac{\omega}{2} u_1,\\  
        \dot v_1&=-\frac{\partial \widetilde{K}_\mu}{\partial u_1}=
        \frac{(v_1^2+v_2^2)u_1}{4(u_1^2+u_2^2)^2}+
        \frac{\omega}{2} v_2-
        \frac{2\mu u_1}{(u_1^2+u_2^2)^2}-\frac{\partial U}{\partial u_1},\\
    \dot v_2&=-\frac{\partial \widetilde{K}_\mu}{\partial u_2}=
    \frac{(v_1^2+v_2^2)u_2}{4(u_1^2+u_2^2)^2}-\frac{\omega}{2} v_1-
    \frac{2\mu u_2}{(u_1^2+u_2^2)^2}-\frac{\partial U}{\partial u_2}.  
    \end{split}
\end{equation*} 
For the next step, we used the following lemma, which was proved in~\cite{Hietarinta:87::}, see also~\cite{Post:2010::,Sergyeyev:12::}.
\begin{lemma}
\label{lem:CCM}
Assume that Hamiltonian generates a Hamiltonian system with $n$ degrees of freedom 
\begin{equation*}
  F(\vq,\vp,\alpha)=F_0(\vq,\vp)-\alpha F_1(\vq,\vp)
\end{equation*}
has a first integral $I(\vq,\vp,\alpha)$ functionally independent of $F$. Then, the system with Hamiltonian 
\begin{equation*}
  G(\vq,\vp,f)=\frac{F_0(\vq,\vp)-f}{F_1(\vq,\vp)}
\end{equation*}
has a first integral $J(\vq,\vp,f)=I(\vq,\vp,G(\vq,\vp,f))$.
\end{lemma}
Let us apply this lemma to the Hamiltonian~\eqref{eq:hamek2}. We have
\begin{equation*}
  \begin{split}
  F_0(\vu,\vv)=&\frac{v_1^2+v_2^2}{8(u_1^2+u_2^2)}+\frac{\omega}{2}(u_2v_1-u_1v_2) +U(u_1,u_2),\\
  F_1(\vu,\vv)=&\frac{1}{4(u_1^2+u_2^2)},
  \end{split}
\end{equation*}
and $\alpha=4\mu$,$f=h$. Denoting $K(\vu,\vv)=G(\vu,\vv,h)$, we obtain
\begin{align}
  \label{eq:Ku}
  K(\vu,\vv)= \frac{1}{2}(v_1^2+v_2^2)+2(u_1^2+u_2^2)\omega(u_2v_1-u_1v_2) \\
  +4(u_1^2+u_2^2)(U(u_1,u_2)-h).
\end{align}

Equations of motion generated by Hamiltonian~\eqref{eq:Ku} admit invariant planes. 
To simplify their forms, we perform the additional
canonical change of the variables
\begin{equation}
\begin{aligned}
	&u_1=\frac{x_1+\rmi x_2}{\sqrt{2}}, && u_2=\frac{\rmi x_1+x_2}{\sqrt{2}},\\
	&v_1=\frac{y_1-\rmi y_2}{\sqrt{2}},&& v_2=\frac{-\rmi y_1+y_2}{\sqrt{2}}.
 \end{aligned}
\end{equation}
After this transformation, the Hamiltonian $K_0$ takes the form
\begin{equation}
	\begin{split}
		\label{eq:Hs}
		\widetilde{K}&=-\rmi(y_1y_2+8h x_1x_2)+4\omega x_1x_2\left(x_2y_2-x_1y_1\right)\\
		& +8\rmi
		x_1x_2\widetilde{V}(x_1,x_2),
	\end{split}
\end{equation}
where $\widetilde{V}(x_1,x_2)=V(x_1^2-x_2^2,\rmi(x_1^2+x_2^2))$.
The corresponding equations of motion are as follows
\begin{equation}
	\begin{split}
		\label{eq:vhs}
	\dot x_1=&-\rmi y_2-4\omega x_1^2x_2,\\
	\dot x_2=&-\rmi y_1+4\omega x_1x_2^2,\\
	\dot y_1=& 4x_2  \omega(2x_1y_1-x_2y_2)  
  + 8\rmi x_2\left(h-\widetilde{V}(x_1,x_2) - x_1\pder{\widetilde{V}(x_1,x_2)}{x_1}\right),\\
		\dot y_2=& 4x_1  \omega(x_1y_1-2x_2y_2)  
  + 8\rmi x_1\left(h-\widetilde{V}(x_1,x_2) - x_2\pder{\widetilde{V}(x_1,x_2)}{x_2}\right).
	\end{split}
\end{equation}
Now, it is evident that the system~\eqref{eq:vhs} possesses two simple
invariant planes, which are given by
\begin{equation}
    \begin{split}
        &\scM_1=\left\{(x_1,x_2,y_1,y_2)\in \C^4\, \big{|}\ x_2=y_1=0\right\},\\
        &\scM_2=\left\{(x_1,x_2,y_1,y_2)\in \C^4\, \big{|}\ x_1=y_2=0\right\}.
    \end{split}
\end{equation}
For further analysis, we restrict the system  to the first plane
\begin{equation}
\begin{aligned}
	\label{eq:vhn1}
	&\dot x_1=-\rmi y_2, \\ & \dot y_2=8\rmi \left[h-\widetilde{V}(x_1,0)\right]x_1.
\end{aligned}
\end{equation}
Knowing that $\widetilde{V}(x_1,0)=V(x_1^2, \rmi x_1^2)$ and $V$ is a sum of two  homogeneous functions $V_k$ and $V_m$, we obtain the differential equation
\begin{equation}
	\label{eq:vhs0}
	\ddot x_1-8\left[h-\lambda_k x_1^{2k}-\lambda_m x_1^{2m}\right]x_1=0,
\end{equation}
where in the last step, we have used the homogeneity property
\begin{equation*}
\begin{split}
&V_i(x_1^2,\rmi x_1^2)=\lambda_i x_1^{2i},\quad \text{where}\qquad \lambda_i:=V_i(1,\rmi
).
\end{split}
\end{equation*}
Assuming that $k,m\not\in\{ -1,0\}$, we find that Eq.~\eqref{eq:vhs0} has 
the first integral of the form
\begin{equation}
	\label{eq:integral}
	I=\frac{1}{2}\dot x_1^2+4x_1^2\left[\frac{\lambda_k}{k+1}x_1^{2k}+\frac{\lambda_m}{m+1}x_1^{2m}-h\right].
\end{equation}
The function $I$ can be treated as the conservation of the energy of the system~\eqref{eq:vhs0}, where $I=e$ is its level. 

Let $\vX=[X_1,Y_2,X_2,Y_1]^T$ denotes the variations of $\vx=[x_1,y_2,x_2,y_1]^T$, then the variational equations
 restricted to $\scM_1$, are as follows
\begin{equation}
	\begin{split}
		\label{eq:war}
	\frac{\rmd}{\rmd \tau}\vX=
	\bold{A}(\tau) \vX,
	\end{split}
\end{equation}  
with a non-constant matrix
\[\quad \bold{A}(\tau)=\begin{bmatrix}
		0,&-\rmi&-4\omega x_1^2&0\\
		a_{12}&0&-8\rmi\omega x_1\dot x_1&4\omega x_1^2\\
		0&0&0&-\rmi\\
		0&0&a_{12}&0
			\end{bmatrix},\]
where
$a_{12}=4\rmi (2h-2(1+2k)\lambda_k x_1^{2k}-2(1+2m)\lambda_m x_1^{2m})$. In the above calculations, we used the Euler identity for homogeneous functions. For instance for $V_k$, we write
\begin{equation}
	x_1\pder{V_k}{x_1}+x_2\pder{V_k}{x_2}=k V_k,
\end{equation}
which enables
\begin{equation}
    \begin{split}
        & x_1^2\pder{V_k}{x_1}(x_1^2,\rmi x_1^2)+\rmi x_1^2\pder{V_k}{x_2}(x_1^2,\rmi
x_1^2)\\ 
&=\left[\pder{V_k}{x_1}(1,\rmi)+\rmi\pder{V_k}{x_2}(1,\rmi)\right]x_1^{2k}=k\lambda_k
x_1^{2k}.
    \end{split}
\end{equation}

Variational equations~\eqref{eq:war} form a system of four first-order
differential equations. For better readability, we rewrite it as a system of 
two
second-order differential equations
\begin{subequations}
\label{eq:var2nd}
\begin{align}\label{eq:var2nda}
	&\ddot X_2+a(\tau)X_2=0,\\ \label{eq:var2ndb}	&\ddot X_1+a(\tau)X_1=b(\tau)X_2,
\end{align}
\end{subequations}
where 
\begin{equation*}
\begin{split}
	&a(\tau)=-8\left[h-(1+2k)\lambda_k x_1^{2k}-(1+2m)\lambda_m x_1^{2m}\right],\\ 
 &b(\tau)=-16\omega\, x_1 \dot x_1.
 \end{split}
\end{equation*}
These coefficients  are functions defined on the hyper-elliptic curve 
\begin{equation}
	\label{eq:curve}
	\dot x_1^2 =2e+8h 
	x_1^2-\frac{8\lambda_k}{k+1}x_1^{2(k+1)}-\frac{8\lambda_m}{m+1}x_1^{2(m+1)}.
\end{equation}
%
%
To simplify further computations, we set  $h=e=0$, and assume
$\lambda_k\lambda_m\neq 0$, that is $V_k(1,\rmi)V_m(1,\rmi)\neq 0$. Thanks to
this, the change of the independent variable
\begin{equation}
\label{eq:change_1}
\tau\to z=1+\frac{(1+k)\lambda_m}{(1+m)\lambda_k}\left(x_1(\tau)\right)^{2(m-k)},
\end{equation}
together with transformation rules for the derivatives
\begin{equation}
\label{eq:chain_rule}
	\begin{split}
		&\frac{\rmd}{\rmd \tau}=\dot z\frac{\rmd}{\rmd z},\qquad \frac{\rmd^2}{\rmd \tau^2}=\dot z^2\frac{\rmd^2}{\rmd z^2}+\ddot z \frac{\rmd}{\rmd z}, 
	\end{split}
\end{equation}
convert the variational equations~\eqref{eq:var2nd} into the following forms
\begin{subequations}
\label{eq:var2rational_rat}
\begin{align}
&\label{eq:var2rational_rat_a}
X_2''+p(z)X_2'+q(z)X_2=0,\\ \label{eq:var2rational_rat_b}
&	X_1''+p(z)X_1'+q(z)X_1=s(z)X_2.
\end{align}
\end{subequations}
Here, $p(z)$, $q(z)$, and $s(z)$ are non-constant rational functions, given by
\begin{equation*}
\label{eq:pp_qq}
\begin{aligned}
p(z)&=\frac{\ddot{z}}{\dot{z}^{2}}
      =\frac{1}{2}\!\left[\frac{1}{z}+\frac{n-l-8}{4(1-z)}\right],\\[4pt]
q(z)&=\frac{a(z)}{\dot{z}^{2}}
      =\frac{1}{32}\!\left[\frac{n^{2}-2nl-15l^{2}}{8(1-z)^{2}}
        +\frac{16+n+11l}{z(1-z)}\right],\\[4pt]
s(z)&=\frac{b(z)}{\dot{z}^{2}}
      =\Omega\,\frac{(1-z)^{\frac{n-4}{2}}}{\sqrt{z}},
      \qquad \Omega\in\C\setminus\{0\},
\end{aligned}
\end{equation*}
where $n,l\in\Q$ are auxiliary parameters previously introduced in~\eqref{eq:n,l}.

Now, we make the classical  Tschirnhaus transformation of dependent variables
\begin{equation}
\label{eq:reduced_change}
\begin{split}
&X_2=X \exp\left[-\frac{1}{2}\int p(z)\rmd z\right],\\
 & X_1=Y\exp\left[-\frac{1}{2}\int p(z)\rmd z\right],
\end{split}
\end{equation}
Thanks to that, we can
rewrite~\eqref{eq:var2rational_rat_a}-\eqref{eq:var2rational_rat_b} to their
reduced forms
\begin{subequations}\label{eq:reduced}
\begin{align}
\label{eq:reduced_a}
X''&=r(z) X,\\ \label{eq:reduced_b}
Y''&=r(z) Y+s(z)X.
\end{align}
\end{subequations}
The coefficients of the above system are 
\begin{equation}
\begin{split}
\label{eq:rr_ss}
r(z)&=\frac{1}{4}\left[\frac{\rho^2-1}{z^2}+
\frac{\sigma^2-1}{(1-z)^2}-\frac{1-\rho^2-\sigma^2+\tau^2}{z(1-z)}\right],\\
s(z)&=\Omega\frac{(1-z)^{\frac{n-4}{2}}}{\sqrt{z}},\quad 
\Omega\in \C\setminus\{0\}
\end{split}
\end{equation}
Here $\rho, \sigma, \tau$ are the differences of the exponents of Gauss
differential equation~\eqref{eq:reduced_a}, with values 
\begin{equation}
\label{eq:delta}
\rho=\frac{1}{2},\quad \sigma=\frac{l}{2},\quad \tau=\frac{3+l}{2}.
\end{equation} 
The respective exponents are given by
\begin{equation*}
  \label{eq:98}
  \rho_{1,2} =\frac{ 1\pm \rho }{2}, \quad \sigma_{1,2} =
 \frac{1\pm \sigma}{2},\quad 
  \tau_{1,2} =\frac{ -1 \pm \tau}{2}.
\end{equation*} 
Equation~\eqref{eq:reduced_a} is reducible as 
\begin{equation}
    -\rho - \sigma +\tau =1,
\end{equation}
see Appendix~B. Its one solution is algebraic, and it has the following form. 
\begin{equation}
\label{eq:X_sol_1}
x_1(z) =z^{3/4} (1-z)^{\frac{2+l}{4 }},
\end{equation}
the second one is given by 
\begin{equation}
	\label{eq:X_sol_2F}
	\begin{split}
	x_2(z)&=x_1(z)\int \frac{1}{x_1(z)^2}\rmd z\\
	&= \sqrt[4]{z} (1-z)^\frac{2+l}{4 } \, F\left(-\frac{1}{2},1+\frac{l}{2 };
  \frac{1}{2};z\right).
	\end{split}
	\end{equation}
Here $F(\alpha,\beta;\gamma;z):={} _2F_1(\alpha,\beta;\gamma;z)$, is the
Gaussian hypergeometric function; for details, see  Appendix~B. 

%
%

%
\subsection{Case $k=2$ and $\lambda_k= 0,$ and  $\lambda_m\neq 0$}
 Let us assume $h\neq 0$. Then  we perform the following change of the independent variable 
\begin{equation}
\tau\longmapsto z=1-\frac{\lambda_m}{(m+1)h}\left(x_1\right(\tau))^{2m},\quad \text{at}\quad e=0.
\end{equation}
This change of variables, together with transformations of
derivatives~\eqref{eq:chain_rule},
convert the system~\eqref{eq:var2nd} to the rational form~\eqref{eq:var2rational_rat}, with the coefficients
\begin{equation}
	\begin{split}
	\label{eq:coeff}
		p(z)&=\frac{\ddot z}{\dot z^2}=\frac{3 z-1}{2 (z-1) z},\\ 
    q(z)&=\frac{a(z)}{\dot
			z^2}=-\frac{m (2 m+3) (z-1)+z}{4 m^2 (z-1)^2 z},\\ 
      s(z)&=\frac{b(z)}{\dot
			z^2}=\frac{\Omega  (1-z)^{\frac{1}{m}-2}}{\sqrt{z}},\quad
		\Omega\in \C\setminus\{0\}.
	\end{split}
\end{equation}
After the Tschirnhaus transformation, we obtain the reduced form of the variational equations~\eqref{eq:reduced}, with the coefficients
\begin{equation}
\begin{split}
\label{eq:rr_ss_k=2_m}
r(z)&=\frac{-3 m^2+(m+2) (5 m+2) z^2-6 (m+2) m z}{16 m^2 (z-1)^2 z^2},\\
s(z)&=\frac{\Omega  (1-z)^{\frac{1}{m}-2}}{\sqrt{z}}.
\end{split}
\end{equation}
With these coefficients, equation~\eqref{eq:reduced_a} is reducible. Its algebraic solution is
\begin{equation}
\label{eq:X_sol_1_k=2_m}
x_1(z)=z^{3/4} (1-z)^{\frac{m+1}{2 m}}.
\end{equation}
The second solution is $x_2(z)=x_1(z)\psi(z)$ where 
\begin{equation}
\label{eq:psi_k=2_m}
\psi(z)=\int \frac{1}{x_1(z)^2}\rmd z = -\frac{2}{\sqrt{z}} \, 
F\left(-\frac{1}{2},1+\frac{1}{m};\frac{1}{2};z\right)
\end{equation}
Moreover, integrals $\varphi(z)$ and $I(z)$ defined in~\eqref{eq:psiphi} take
the forms
\begin{equation}
\label{eq:phi_k=2_m}	
\begin{split}
\varphi(z)=&-\frac{m \Omega  (1-z)^{2/m} (m+2 z)}{2 (m+2)},\\
I(z)=&\int \varphi(z)\psi(z)\rmd z
= a(z)+b\sqrt{z}F\left(\frac{1}{2},1 -\frac{1}{m};\frac{3}{2};z\right)
\end{split}
\end{equation}
where $a(z)$ is an algebraic function and $b$ is a non-zero constant (their
explicit forms are irrelevant for our further considerations).
%
\subsection{Case $k=2$ and $\lambda_k\neq 0,$ and  $\lambda_m=0$}
 Let us assume $h\neq 0$. Then  we perform the following
change of the independent variable 
\begin{equation}
\tau\to z=\sqrt{1-2\sqrt{\frac{h}{\lambda_k}}x_1^{-2}(\tau)},\quad \text{at}\quad e=\frac{8}{3}\sqrt{\frac{h^3}{\lambda_k}}.
\end{equation}
This change of variables, combined with the derivative transformations given in~\eqref{eq:chain_rule},
recasts the system~\eqref{eq:var2nd} into the rational form~\eqref{eq:var2rational_rat}, with the coefficients \begin{equation}
	\begin{split}
	\label{eq:coeff2}
		p(z)&=\frac{\ddot z}{\dot z^2}=\frac{2z}{z^2-3},\\ q(z)&=\frac{a(z)}{\dot
			z^2}=\frac{3(z^4-2z^2-19)}{(z^2-1)^2(z^2-3)^2},\\ s(z)&=\frac{b(z)}{\dot
			z^2}=\frac{\Omega z}{(z^2-3)(z^2-1)^2},\quad
		\Omega\in \C\setminus\{0\}.
	\end{split}
\end{equation}
Performing the Tschirnhaus transformation of dependent variables~\eqref{eq:reduced_change}, we obtain the reduced form of the variational equations~\eqref{eq:reduced}, with the coefficients
\begin{equation}
\begin{split}
\label{eq:rr_ss_k=2}
r(z)&= -\frac{6(z^4-2z^2-9)}{(z^2-1)^2(z^2-3)},\\
s(z)&=\frac{3\Omega z}{(z^2-3)(z^2-1)^2}.
\end{split}
\end{equation}
With these coefficients, equation~\eqref{eq:reduced_a} is reducible. Its algebraic solution is
\begin{equation}
\label{eq:X_sol_1_k=2}
x_1(z)=\frac{z(z^2-3)}{(z^2-1)^{3/2}}.
\end{equation}
The second solution is $x_2(z)=x_1(z)\psi(z)$ where
\begin{align}
\label{eq:psi_kalign=2}
\psi(z)=\int \frac{1}{x_1(z)^2}\rmd z = \frac{-4z^4+11z^2-3}{9z(z^2-3)^2}\\-\frac{5}{9\sqrt{3}} \operatorname{arctanh}\left(\frac{z}{\sqrt{3}}\right)
\end{align}
Moreover, integrals $\varphi(z)$ and $I(z)$ defined in~\eqref{eq:psiphi} take the forms
\begin{equation}
\label{eq:phi_k=2}	
\begin{split}
\varphi(z)=&\Omega\frac{ \left(-2 z^6+9 z^4-12 z^2+3\right)}{4 \left(z^2-1\right)^4},\\
I(z)=&\int \varphi(z)\psi(z)\rmd z
=\frac{1}{288} \Omega  \left[\frac{3 \left(z^4+7 z^2-24\right)}{z \left(z^2-3\right)^2} \right.\\
& \left.-18
    \operatorname{arctanh}(z)+5 \sqrt{3} \operatorname{arctanh} \left(\frac{z}{\sqrt{3}}\right)\right].
\end{split}
\end{equation}
%
%
\
\section{Proofs of Theorems~\ref{th:main_theorem}--\ref{th:main_theorem_k=2}}
\label{sec:proof_main_theorem}
The proof of Theorem~\ref{th:main_theorem} is based on the following lemma.
\begin{lemma}
\label{lem:necsu}
For $l,n\in \Q$, the  identity component of the differential Galois group of the
system~\eqref{eq:reduced} is Abelian if and only if 
\begin{enumerate}
  \item \label{it:necs1} $l\geq-1$ is  an odd, or $l<-1$ is an even integer,
   or 
  \item \label{it:necs2} $(n+l)(n+l+2)\neq 0$ and either 
  \begin{enumerate}
    \item $n>0$ is an even, or $n<0$ is an odd integer; or
  \item \label{it:necs3} $n+l$ is an even integer, except the case when $l\geq 0$ and $n\leq 0$ are both even integers.
  \end{enumerate}
\end{enumerate}
\end{lemma}
We will prove it in the next section. 

First, we show the following fact.
\begin{theorem}
  \label{thm:Kmu}
  Under the assumptions of Theorem~\ref{th:main_theorem}, if the system governed by Hamiltonian
  \eqref{eq:Ku} is integrable in the Liouville sense with meromorphic first integrals, then:
\begin{enumerate}
  \item \label{it:necs1} $l\geq-1$ is  an odd, or $l<-1$ is an even integer,
   or 
  \item \label{it:necs2} $(n+l)(n+l+2)\neq 0$ and either 
  \begin{enumerate}
    \item $n>0$ is an even, or $n<0$ is an odd integer; or
  \item \label{it:necs3} $n+l$ is an even integer, except the case when $l\geq 0$ and $n\leq 0$ are both even integers.
  \end{enumerate}
\end{enumerate}
\end{theorem}
\begin{proof}
  The system governed by Hamiltonian \eqref{eq:Ku} has a family of particular solutions defined by \eqref{eq:vhn1}, or equivalently by \eqref{eq:vhs0}, and the corresponding variational equations are given by~\eqref{eq:var2nd}. If it is integrable, then according to the Morales-Ramis Theorem~\ref{th:MR}, the identity component of the differential Galois group of the variational equations along the mentioned particular solution is Abelian. Invoking Lemma~\ref{lem:necsu}, we obtain the thesis of the theorem. 
\end{proof}
Now, the proof of Theorem~\ref{th:main_theorem} is simple.
\begin{proof}{[Proof of Theorem~\ref{th:main_theorem}]} 
Let us assume that the system defined  Hamiltonian~\eqref{eq:Hamek} is integrable. Then it admits an additional  first integral, which is a meromorphic function of $(q_1,q_2,p_1,p_2,r)$. After the Levi-Civita transformation, it will be a meromorphic function of $(u_1,u_2,v_1,v_2)$ which is an additional first integral system defined by Hamiltonian~\eqref{eq:hamek2}. Thus, by 
by Lemma~\ref{lem:CCM}, the system governed by Hamiltonian \eqref{eq:Ku} is integrable. But 
 by 
 Theorem~\ref{thm:Kmu} states that it is not integrable. The contradiction finishes the proof.  
\end{proof}
We pass now to the proof of Theorem~\ref{th:main_theorem_k=2_m}. The proof is based on the following lemma which we prove in Section~\ref{sec:proof_lemma_m}.
\begin{lemma}
\label{lem:necsu_m}
If $\abs{m}>2$ then  the  identity component of the differential Galois group of the
system~\eqref{eq:reduced} with coefficients defined by~\eqref{eq:rr_ss_k=2_m} is not Abelian.  
\end{lemma}
With the above lemma, the proof of Theorem~\ref{th:main_theorem_k=2_m} is similar as the proof of Theorem~\ref{th:main_theorem} and we left it to the reader.

Proof of Theorem~\ref{th:main_theorem_k=2} needs more effort.
\begin{proof}{[Proof of Theorem~\ref{th:main_theorem_k=2}]}
 As in the proof of Theorem~\ref{th:main_theorem}, first we investigate integrability  the system defined by
 Hamiltonian~\eqref{eq:Hamek}. To this end  is integrable, variational equations along a particular
 solution defined by \eqref{eq:vhn1}. 
They are given by~\eqref{eq:reduced}, with
coefficients defined by~\eqref{eq:rr_ss_k=2}. To show that the identity
 component of its differential Galois group is not Abelian, we apply
Lemma~\ref{lem:necsu2}. Taking. into account that for the considered case
functions $\psi(z)$ and $I(z)$ are given by~\eqref{eq:psi_k=2}
and~\eqref{eq:phi_k=2}, we find that function $g(z)$ defined by~\eqref{eq:gzA2}
has the following form 
\begin{equation*}
\label{eq:g_k=2}
g(z)= R(z) +(\lambda c_1 + c_2)  \operatorname{arctanh}\left(\frac{z}{\sqrt{3}}\right) +c_3 \operatorname{arctanh}(z)
\end{equation*}
where $R(z)$ is a rational function, and $c_1$, $c_2$ and $c_3$ are non-zero constants. For arbitrary $\lambda\in\C$, this function is not algebraic. 
To see this, we note that if $g(z)$ is algebraic, then the function
\begin{equation*}
\widetilde{g}(z)=a_1 \operatorname{arctanh}\left(\frac{z}{\sqrt{3}}\right) +a_2 \operatorname{arctanh}(z),
\end{equation*}
is algebraic for some $a_1,a_2\in\C$, not both zero. This is impossible because
an algebraic function does not have irregular singularities. If $a_1\neq0$ then
$ \operatorname{arctanh}(z/\sqrt{3})$ as well as $\widetilde{g}(z)$ has two irregular singular
points at $z=\pm \sqrt{3}$. Thus, necessarily $a_1=0$. In our case $a_1= \lambda
c_1 + c_2$, so we have to fix $\lambda= -c_2/c_1$. However, if $a_2\neq0$, then
$ \operatorname{arctanh}(z)$  and $\widetilde{g}(z)$ has two irregular singular points at
$z=\pm 1$. Thus, necessarily $a_2=0$. This is a contradiction which proves our
statement. Thus, the identity component of the differential Galois group of the
variational equations is not Abelian, which finishes the proof of
Theorem~\ref{th:main_theorem_k=2}. 
\end{proof}
\section{Main lemma}
\label{sec:main_lemma}
In this section, we formulate and prove the Lemma~\ref{lem:necsu}  which plays a key role in the
proof of Theorem~\ref{th:main_theorem}. It specifies all cases when the identity
component of the system's differential Galois group of the
system~\eqref{eq:reduced} is Abelian.
%
%

To prove this lemma, we use the criteria formulated in Appendix~A. To this end
we must compute the integrals $\psi(z)$, $\varphi(z)$, and $I(z)$ defined
in~\eqref{eq:psiphi}. The functions $r(z)$ and $s(z)$ for the system
\eqref{eq:reduced} are given by~\eqref{eq:rr_ss}. The integral $\psi(z)$ is
\begin{equation}
  \label{eq:psi}
  \psi(z)= -\frac{2}{\sqrt{z}}\,F\left(-\frac{1}{2},\,1+\frac{l}{2};\,\frac{1}{2};\,z\right).
\end{equation}
If this function is algebraic, then the two solutions of
equation~\eqref{eq:reduced_a} are algebraic, and the identity component of the
differential Galois group of system~\eqref{eq:reduced} is Abelian; see
Lemma~\ref{lem:necsu0}. Note that $\psi(z)$ is algebraic if and only if the
solution $x_2(z)$ given by~\eqref{eq:X_sol_2F} is algebraic.

We show first the following facts.
\begin{proposition}
  \label{prop:reduced}
  Solution~\eqref{eq:X_sol_2F} is algebraic if and only if either $l\geq -1$ is
  an odd integer, or $l<-1$ is an even integer. In this case, the
  identity component of the differential Galois group of the 
  system~\eqref{eq:reduced} is Abelian.
\end{proposition}
\begin{proof}
  If  solution~\eqref{eq:X_sol_2F} is algebraic, then it has the form 
  \begin{equation}
    \label{eq:x2a}
    x_2(z)= z^{e_0}(z-1)^{e_1}p(z) 
  \end{equation}
  where $p(z)$ is a polynomial of degree $d=-(e_0+e_1+e_\infty)$, and
  $e_0$,$e_1$ and $e_\infty$ are exponents at the points $0$, $1$, and $\infty$
  respectively, see~\cite[Ch.4]{Iwasaki:91::}. From this fact, we deduce that
  $l$ has to be an integer. This is only a necessary condition. If $l$ is an
  integer, then solution~\eqref{eq:X_sol_2F} can have a logarithmic term. The
  necessary and sufficient condition for the presence of logarithmic term is
  given in~\cite[Lemma 4.3.7]{Iwasaki:91::}. Using them, we find they do not
  appear if and only if $l$ satisfies the given assumptions.  

  If both solutions of equation~\eqref{eq:reduced_a} are algebraic, then, as the
  equation is reducible, its differential Galois group, as well as the monodromy
  group, is contained in the diagonal subgroup of $\mathrm{SL}(2,\C)$. Then by
  Theorem~3.2 in \cite{Duval:09::}, the identity component of the differential
  Galois group of the system~\eqref{eq:reduced} is Abelian.  
\end{proof}

The form of integral $\varphi(z)$  defined by~\eqref{eq:psiphi} depends on the
values of $n$ and $l$. Namely, if $(n+l) (n+l+2)\neq 0$, then $\varphi(z)$ is
algebraic
\begin{equation}
  \varphi(z) = -2\Omega\frac{  [2+z (n+l)](1-z)^{\frac{n+l}{2}} }{(n+l) (n+l+2)}.
\end{equation}
In this case, we need also the integral $I(z)$, which is given by
\begin{equation}
  \label{eq:I}
  I(z) = 
	a(z) + b \sqrt{z} F\left(\frac{1}{2},1-\frac{n}{2};\frac{3}{2};z\right),  
\end{equation}
where $a(z)$ is an algebraic function and $b$ is a non-zero complex number
(The explicit form of these coefficients is irrelevant for further
calculations).

The case when $(n+l) (n+l+2)= 0$ correspond to $k=-2$ or $m=-2$. 
If $l=-n$, then
\begin{equation}
  \label{eq:phi1}
  \varphi(z)=-a (z+\log (z-1))
\end{equation}
and if $l=-n-2$, then
\begin{equation}
  \label{eq:phi2}
  \varphi(z)=a \left(\log (z-1)-\frac{1}{z-1}\right).
\end{equation}
In both cases, function $\varphi(z)$ is not algebraic.



  


The proof of Lemma~\ref{lem:necsu} depends  on the condition   $(n+l)
(n+l+2)\neq 0$. 

\subsection{The algebraic case }
\label{ssec:algebraic}

This section assumes that $(n+l) (n+l+2)\neq 0$.  

As for cases considered in this section, $\varphi(z)$ is algebraic, we will use
criterion given in Lemma~\ref{lem:necsu2}. Since integrals $\psi(z)$ and $I(z)$
are given by~\eqref{eq:psi} and~\eqref{eq:I}, it is clear that, after the
rearrangement of terms, as a function $g(z)$ in Lemma~\ref{lem:necsu2} we can
take 
\begin{equation}
	\label{eq:is alg}
	g(z):= 
		\lambda F(\alpha,\beta;\gamma;z)+
  zF(\widehat\alpha,\widehat\beta;\widehat\gamma;z),
\end{equation} 
where  
\begin{equation}
	\label{eq:para0}
	\begin{split}
		(\alpha,\beta,\gamma) = &\left(-\frac{1}{2},1+\frac{l}{2 },\frac{1}{2}\right),\\
		(\widehat\alpha,\widehat\beta,\widehat\gamma) = &
		\left(\frac{1}{2},1-\frac{n}{2},\frac{3}{2}\right).
	\end{split}
\end{equation}
We have to check if there exists $\lambda\in\C$ such that $g(z)$ is
algebraic. Let us notice that if $F(\widehat\alpha,\widehat\beta;\widehat\gamma;z)$ is algebraic, then  $g(z)$ is algebraic as well for $\lambda =0$. In these cases, we distinguish by the following proposition.
\begin{proposition}
\label{prop:Fn}
  The hypergeometric function $F\left(\tfrac{1}{2},1-\tfrac{n}{2};\tfrac{3}{2};z\right)$ is algebraic if and only if  either $n$ is positive and even, or $n$ is negative and odd.
\end{proposition}
The proof of this proposition is similar to the proof 
Proposition~\ref{prop:reduced}, so we omit it.

To verify if $g(z)$ is algebraic, we investigate its analytic
continuations along closed paths with a common one point. If it is an algebraic
function, then such continuations can give only a finite number of different
values. As $g(z)$ is a linear combination of two hypergeometric functions, we have to analyse the analytical continuations of both of them. 
Thus, we have to analyse two  Gauss hypergeometric equations with respective  parameters
$(\alpha,\beta,\gamma)$ and
$(\widehat\alpha,\widehat\beta,\widehat\gamma)$. The respective bases of
local solutions  in a neighbourhood of  singularity $z=0$ are 
\begin{align*}
  u_1(z) = & F(\alpha,\beta;\gamma;z),                               \qquad\qquad
  u_2(z)= \sqrt{z},  \\ \widehat{u}_1(z) =
           & F(\widehat\alpha,\widehat\beta;\widehat\gamma;z), \qquad\qquad
  \widehat{u}_2(z)=  \frac{1}{\sqrt{z}}. 
\end{align*}
We shall check how these local solutions change during analytical continuation
along certain closed loops. Notions of analytical continuation and local and
global monodromy and their calculations for the hypergeometric equation are
presented in Appendix B.

We take two loops  $ \sigma_0$ and $\sigma_1$  with one common point $z_0$
encircling  counter-clockwise singularities $z=0$ and $z=1$, respectively,  see
Fig.~\ref{fig:mono} in Appendix A. By $M_{\sigma_0} $ and $M_{\sigma_1}$ we denote the
respective monodromy matrices. 

Let us assume that neither $l$ nor $n$ is an even integer. Then the
An explicit form of monodromy matrices is defined by~\eqref{eq:M_zero} and
\eqref{eq:10}. For further considerations, we take  the commutator loop
\begin{equation*}
  \rho_1 =
  \sigma_0\sigma_1\sigma_0^{-1}\sigma_1^{-1}
\end{equation*}
and the
corresponding monodromy matrix
\begin{equation}
  \label{eq:11} C := M_{\rho_1}= M_{\sigma_1}^{-1} M_{\sigma_0}^{-1}
  M_{\sigma_1} M_{\sigma_0} 
\end{equation}
We will also need the following commutator%
\begin{equation}
  \label{eq:12} D:=
  M_{\rho_{\infty}}=M_{\sigma_\infty}^{-1} M_{\sigma_0}^{-1}
	M_{\sigma_\infty} M_{\sigma_0}  
\end{equation}

For the sets of parameters~\eqref{eq:para0}, the respective
commutator matrices are denoted by $C$, $\widehat C$ and $D$, $\widehat D$.
All of them are unipotent and lower triangular, thus they have the form
\[
\begin{aligned}
  C&=\begin{bmatrix}
    1 & 0 \\
    c_{21} & 1
  \end{bmatrix},\quad
  \widehat C=\begin{bmatrix}
    1 & 0 \\
    \widehat c_{21} & 1
  \end{bmatrix},\\[1em]
  D&=\begin{bmatrix}
    1 & 0 \\
    d_{21} & 1
  \end{bmatrix},\quad
  \widehat D=\begin{bmatrix}
    1 & 0 \\
    \widehat d_{21} & 1
  \end{bmatrix}.
\end{aligned}
\]
The analytical continuations of $F(\alpha,\beta;\gamma;z)$ and
  $F(\widehat\alpha,\widehat\beta;\widehat\gamma;z)$ along loop $\rho_1$ give
\begin{equation*}
  \begin{split}
    \cM_{\rho_1}( F(\alpha,\beta;\gamma;z)) = &
    F(\alpha,\beta;\gamma;z) +c_{21} \sqrt{z},\\
    \cM_{\rho_1}(F(\widehat \alpha,\widehat \beta;\widehat \gamma;z))=& 
    F(\widehat \alpha,\widehat \beta;\widehat \gamma;z)
    +\widehat {c}_{21}\frac{1}{\sqrt{z}}.
  \end{split}
\end{equation*}
Hence
\begin{align*}
  \cM_{\rho_1}(g(z)) = \lambda \cM_{\rho_1}( F(\alpha,\beta;\gamma;z)) + z
  \cM_{\rho_1}(F(\widehat \alpha,\widehat \beta;\widehat \gamma;z))\\= g(z) +\sqrt{z}\Delta_1,
\end{align*}
where
\begin{equation*}
  \Delta_1 :=  \lambda c_{21} + \widehat  c_{21}.
\end{equation*}
If $\Delta_1\neq0$ then continuation along loops $\rho_1^{N}$ give finitely many
  values of $g(z)$ as we have
\begin{equation*}
  \cM_{\rho}(g(z)) = g(z) +N\sqrt{z}\Delta_1, \quad\text{for}\quad N\in \Z.
\end{equation*}
Thus, if $g(z)$ is algebraic, then $\Delta_1=0$. But this one condition is not
sufficient. In fact, $\Delta_1=0$ for $\lambda= -\widehat c_{21}/ c_{21}$.
This is why we should consider the commutator $D$. Similar reasoning with
analytical continuation along loop $\rho_{\infty}$ gives $M_{\rho_\infty}(g(z)) = g(z) +\sqrt{z}\Delta_\infty$, where $
\Delta_{\infty}:= \lambda d_{21} + \widehat d_{21}$. 

Summarizing, if $g(z)$ is algebraic, then 
$\Delta_1=0$ and $\Delta_{\infty}=0$, that is, there exists $\lambda\in\C$ such that 
\begin{equation*}
      \lambda c_{21} + \widehat c_{21} =0 , \qquad 
      \lambda d_{21} + \widehat d_{21} = 0.
\end{equation*}
It is possible if and only if the following
determinant
\begin{equation}
  \label{eq:20}
  \Delta=  \det
    \begin{bmatrix}
      c_{21} & \widehat c_{21} \\ d_{21} & \widehat d_{21}
    \end{bmatrix}
\end{equation}
vanishes.
The explicit form of the non-trivial elements of the commutator matrices
  $C$, $\widehat C$, $D$, and $\widehat D$ are following 
\begin{equation*}
  \label{eq:c21d21}
    c_{21}= \rme^{-\rmi \pi l}d_{21}, \qquad
    \widehat c_{21}= -\frac{\sqrt{\pi } \left(1-\rme^{\rmi
          n  \pi }\right) \mathrm{\Gamma}
      \left(\frac{n}{2}\right)}{\mathrm{\Gamma} \left(\frac{1+n}{2}\right)}, 
\end{equation*}
\begin{equation*}
    d_{21}=\frac{2 \sqrt{\pi } \left(1-\rme^{ \rmi \pi l}\right)
      \mathrm{\Gamma} \left(-\frac{l}{2}\right)}{\mathrm{\Gamma} \left(-\frac{l+1}{2
          }\right)}, 
   \qquad {\widehat {d}_{21}} = \rme^{-\rmi \pi n}\, \widehat  c_{21}. 
\end{equation*}  
With these formulae, we obtain 
\begin{equation}
  \label{eq:det}
  \Delta = R \rme^{-\rmi \pi  (l+n)} \left(\rme^{\rmi \pi  l}-1\right) \left(\rme^{\rmi \pi  n}-1\right)
    \left(\rme^{\rmi \pi  (l+n)}-1\right),
\end{equation}
where $R$ is  given by     
\begin{equation}
  \label{eq:R}
  R = -\frac{2 \pi  \mathrm{\Gamma}
   \left(-\frac{l}{2}\right) \mathrm{\Gamma} \left(\frac{n}{2}\right)}{\mathrm{\Gamma}
   \left(-\frac{l+1}{2}\right) \mathrm{\Gamma}
   \left(\frac{n+1}{2}\right)}.
\end{equation}
%
By our assumptions, $R\neq 0$. Hence, if $g(z)$ is algebraic, then $\Delta = 0$.
Since neither $l$ nor $n$ is an even integer, from \eqref{eq:det} we deduce that
$\Delta = 0$ if and only if $n+l$ is an even integer. This completes the proof
of Lemma~\ref{lem:necsu} in the case $(n+l)(n+l+2)\neq 0$ with neither $n$ nor
$l$ being an even integer.

We have to investigate the cases when $l$ or $n$ is an even integer. At first,
let us assume that $l=2l'$ is an even integer. Moreover, we assume also that  $n$
is neither an even positive nor a negative and odd integer. This guarantees that
the function
$\widehat{u}_1(z)=F\left(\tfrac{1}{2},1-\tfrac{n}{2};\tfrac{3}{2};z\right)$ is
not algebraic, see Proposition~\ref{prop:Fn}. 

If $l=2l'<-1$, then the statement 1 of Lemma~\ref{lem:necsu} follows from
Proposition~\ref{prop:reduced}. Thus, we assume that $l=2l'\geq 0$ is an even
integer, so $(\alpha,\beta,\gamma)=(-\tfrac{1}{2},1+l',\tfrac{1}{2})$. For these parameters  the respective local monodromy matrix at $z=1$ is $\widetilde M_{\sigma_1} =\begin{bsmallmatrix}
    1 & 0 \\ 2\pi \rmi& 1
  \end{bsmallmatrix}$. Hence, for any non-zero integer $\nu$ we have ${\widetilde M_{\sigma_1}}^\nu =\begin{bsmallmatrix}
    1 & 0 \\ 2\pi \rmi \nu & 1
  \end{bsmallmatrix}$, see~\eqref{eq:15}.
As $ M_{\sigma_1} =P^{-1} \widetilde M_{\sigma_1} P$, where $P$ is given by~\eqref{eq:16}, we have 
\begin{equation}
  \label{eq:M_one_nu}
  M_{\sigma_1}^\nu = P^{-1} {\widetilde M_{\sigma_1}}^\nu P =
  \id + \Delta S, \quad \Delta = 2\pi \rmi \nu,
\end{equation}
and  matrix $S$ is given by
\begin{equation}
  \label{eq:S}
  S=[s_{ij}]= \begin{bmatrix}
    s  & p_{12}^{-1}\\
   -p_{12} s^2 & -s 
  \end{bmatrix}, \quad s= (2-\ln(4)),
\end{equation}
see \eqref{eq:16}.

As $n$ is not an even integer, the local monodromy of the hypergeometric equation with
parameters $(\widehat\alpha,\widehat\beta,\widehat\gamma)$, at singularity $z=1$
is $\widetilde M_{\sigma_1}=\diag(1,\rme^{\rmi\pi n})$ see~\eqref{eq:M_one}.
Because $n$ is a rational number, there exists a positive integer $\nu$ such that
$n \nu$ is an even integer. Then $ {\widetilde M_{\sigma_1}}^\nu=\id$, and so
the global monodromy matrix ${\widehat M_{\sigma_1}}$ satisfies ${\widehat
M_{\sigma_1}}^\nu=\id$.

Let $\rho := \sigma_1^\nu$ be the loop that winds $\nu$ times around the
singularity at $z=1$. The analytical continuations of
$u_1(z)=F(\alpha,\beta;\gamma;z)$ and
$\widehat{u}_1(z)=F(\widehat\alpha,\widehat\beta;\widehat\gamma;z)$ along loop
$\rho$ give
\begin{equation}
  \label{eq:monodromy_rho_l}
    \begin{split}
      \cM_{\rho}( u_1(z)) = &
      u_1(z) +\Delta \left[s_{11} u_1(z)
      +s_{21}\sqrt{z}\right],\\
      \cM_{\rho}(\widehat{u}_1(z))=& \widehat{u}_1(z).
    \end{split}
\end{equation}
Hence
\begin{align*}
  \cM_{\rho}(g(z)) = \lambda \cM_{\rho}( F(\alpha,\beta,\gamma;
  z)) + z
  \cM_{\rho}(F(\widehat \alpha,\widehat \beta,\widehat \gamma;
  z))\\= g(z) +\sqrt{z}\Delta_1,
\end{align*}
Note that we do not take $\lambda=0$ because $u_1(z)$ is not an algebraic function,
see remark after Lemma~\ref{lem:necsu2}. Continuation along loop $\rho^N$, where
$N$ is arbitrary integer,  gives 
\begin{equation*}
  \cM_{\rho^N}(g(z)) =g(z) +\lambda N\Delta \left[s_{11} u_1(z)+s_{21}\sqrt{z}\right].
\end{equation*}
Thus, continuations of the function $g(z)$ can give an arbitrary number of different values, so it is not algebraic.

Because $l$ is not an even integer,  the local monodromy matrix of the
hypergeometric equation with parameters $(\alpha,\beta,\gamma)$ at singularity
$z=1$ is $\widetilde M_{\sigma_1}=\diag(1,\rme^{-\rmi\pi l})$,
see~\eqref{eq:M_one}. As $l$ is a rational number, there exists a positive
integer $\nu$ such that $l \nu$ is an even integer. Then $ {\widetilde
M_{\sigma_1}}^\nu=\id$, and so the global monodromy matrix $ M_{\sigma_1}$
satisfies $M_{\sigma_1}^\nu=\id$.

For parameters $(\widehat{\alpha},\widehat{\beta},\widehat{\gamma})$, the
respective local monodromy matrix at $z=1$ is $\widetilde M_{\sigma_1}
=\begin{bsmallmatrix} 1 & 0 \\ 2\pi \rmi& 1 \end{bsmallmatrix}$. Hence,
${\widetilde M_{\sigma_1}}^\nu =\begin{bsmallmatrix} 1 & 0 \\ 2\pi \rmi \nu
& 1 \end{bsmallmatrix}$, see~\eqref{eq:15}. As $ M_{\sigma_1} =P^{-1}
\widetilde M_{\sigma_1} P$, where $P$ is given by~\eqref{eq:17}, we have 
\begin{equation}
  \label{eq:M_one_nu}
  M_{\sigma_1}^\nu = P^{-1} {\widetilde M_{\sigma_1}}^\nu P =
  \id + \Delta S, \quad \Delta = 2\pi \rmi \nu,
\end{equation}
and  matrix $S$ is given by
\begin{equation}
  \label{eq:S}
  S=[\widehat{s}_{ij}]= \begin{bmatrix}
    s  & p_{12}^{-1}\\
   -p_{12} s^2 & -s 
  \end{bmatrix}, \quad s= -2\ln(2),
\end{equation}
see \eqref{eq:17}. 
As in the previous case, 
we take a loop $\rho =\sigma_1^\nu$ encircling $\nu$ times singularity $z=1$.
The analytical continuations of $u_1(z)=F(\alpha,\beta;\gamma;z)$ and
  $\widehat{u}_1(z)=F(\widehat\alpha,\widehat\beta;\widehat\gamma;z)$ along loop $\rho$ give
\begin{equation*}
    \begin{split}
      \cM_{\rho}( u_1(z)) = & 
      u_1(z) ,\\
      \cM_{\rho}(\widehat{u}_1(z))=& \widehat{u}_1(z) +\Delta \left[\widehat{s}_{11} \widehat{u}_1(z)
      +\frac{\widehat{s}_{21}}{\sqrt{z}}\right].
    \end{split}
\end{equation*}
Hence
\begin{equation*}
  \cM_{\rho}(g(z)) = g(z) +\Delta \left[\widehat{s}_{11} u_1(z)+\frac{\widehat{s}_{21}}{\sqrt{z}}\right].
\end{equation*}
Hence, by the same arguments as in the previous case, the function $g(z)$  is not
algebraic. 

It is left to investigate the case when $l=2l'\geq 0$ and $n=-2n'\leq0$ are both
even integers. We consider continuation along loop $\rho=\sigma_1^\nu$, where $\nu$ is an arbitrary integer, Just using our
above reasoning, we get 
\begin{equation*}
    \begin{split}
      \cM_{\rho}( u_1(z)) = &
      u_1(z) +\Delta \left[s_{11} u_1(z)
      +s_{21}\sqrt{z}\right],\\
      \cM_{\rho}(\widehat{u}_1(z))=& \widehat{u}_1(z) +\Delta \left[\widehat{s}_{11} \widehat{u}_1(z)
      +\frac{\widehat{s}_{21}}{\sqrt{z}}\right].
    \end{split}
\end{equation*}
where $\Delta=2\pi \rmi \nu$. Hence
\begin{align*}
  \cM_{\rho}(g(z)) &= g(z) \\ & +  \Delta \left[\lambda\left( s_{11} u_1(z)+s_{21}\sqrt{z}\right)+\widehat{s}_{11} z\widehat{u}_1(z)+\frac{\widehat{s}_{21}z}{\sqrt{z}}\right].
\end{align*}
If $g(z)$ is algebraic, then the function in the square bracket has to vanish identically. Here, the difficulty  is connected with the fact that it depends on three parameters $l'$, $n'$, and~$\lambda$.

Let us consider the following function
\begin{equation*}
  h(z) = c_1 u_1(z) + c_2 \sqrt{z} + \widehat{c}_1z\widehat{u}_1(z) + \widehat{c}_2 \frac{z}{\sqrt{z}},
\end{equation*}
where $c_1,c_2,\widehat{c}_1,\widehat{c}_2$ are complex constants. We want to find constants that $h(z)$ vanishes identically. Expanding $h(z)$ into the Puiseux  series around $z=0$, we get
\begin{equation*}
  h(z) = c_1 + (c_2 + \widehat{c}_2) z^{\frac{1}{2}}  + \ldots,
\end{equation*}
In our case 
\begin{equation*}
  c_1 = \lambda s_{11} =-\lambda(2-\ln(4)) \dfrac{(-1)^{l'}\sqrt{\pi } }{l'! \mathrm{\Gamma} \left(-l'-\frac{1}{2}\right)}. 
  \end{equation*}
Because $\lambda\neq 0$, and $l'$ is an integer,  we have $c_1\neq 0$. Thus,
$h(z)$ does not vanish identically, and so $g(z)$ is not algebraic for arbitrary
$\lambda\neq 0$ and arbitrary non-negative $l'$ and $n'$. This finishes the
proof of Lemma~\ref{lem:necsu} for the case when $(n+l) (n+l+2)\neq 0$.
\subsection{The logarithmic case}
\label{ssec:logarithmic}
Now we consider the cases when $(n+l) (n+l+2)= 0$. If   the identity component
of the differential Galois group of the system~\eqref{eq:reduced} is  Abelian,
then  According to Lemma~\ref{lem:necsu}, we have to check cases if there exist
$\lambda\in\C$ such that function $\varphi(z)+\lambda\psi(z)$ is algebraic. In
consider cases function $\varphi(z)$ has the form  given by~\eqref{eq:phi1}, or
by~\eqref{eq:phi2}, and functions $\psi(z)$ is defined in~\eqref{eq:psi}, that is 
\begin{equation*}
  \psi(z)=
  -\frac{2 }{\sqrt{z}}F\left(-\frac{1}{2},1+\frac{l}{2 };\frac{1}{2};z\right) = -\frac{2 }{\sqrt{z}} u_1(z) .
\end{equation*}
As already mentioned, independently of the choice of $\varphi(z)$, we have to
check if there exist $\lambda\in\C$ such that function
\begin{equation}
  \label{eq:phi1+psi}
  g(z)=\ln(1-z)+\lambda \psi(z)    
\end{equation}
is algebraic. We have to consider only those values of $l$ which do not satisfy
condition  1 of Lemma~\ref{lem:necsu}.     

We consider the continuation of this function along loops around $z=1$. Obviously,
\[\cM_{\sigma_1}(\ln(1-z)) = \ln(1-z) + 2\pi \rmi,\] and \[\cM_{\sigma_1}(\psi(z))
= -\tfrac{2 }{\sqrt{z}}\cM_{\sigma_1}(u_1(z)).\] But $\cM_{\sigma_1}(u_1(z))$ we
already investigated. So, if $l$ is not an even integer, then the local
monodromy matrix at $z=1$ is diagonal
$\widetilde{M}_{\sigma_1}=\diag(1,\rme^{-\rmi\pi l})$, see~\eqref{eq:M_one}.
Because $l$ is a rational number, there exists a positive integer $\nu$ such that
$l \nu$ is an even integer. Then $ {\widetilde M_{\sigma_1}}^\nu=\id$. Thus, we
take a loop $\rho=\sigma_1^\nu$ and we obtain $\cM_{\rho}(\psi(z)) = \psi(z)$.
In result
\begin{equation*}
  \cM_{\rho}(g(z)) = g(z) + 2\pi \rmi \nu, \qquad \nu\in \Z.
\end{equation*}
We conclude that $g(z)$ is not algebraic. 
   
If $l=2l'\geq 0$ is an even integer, then 
\begin{equation*}
  \cM_{\rho}( u_1(z)) = 
      u_1(z) +\Delta \left[s_{11} u_1(z)
      +s_{21}\sqrt{z}\right],
\end{equation*}
where $\Delta=2\pi\mathrm{i}\nu$, see \eqref{eq:monodromy_rho_l}. 
Consequently
\begin{equation*}
  \cM_{\rho}(\psi(z)) = \psi(z) + \Delta\left[s_{11}\psi(z) -2 s_{21}\right],
\end{equation*}
and 
\begin{equation*}
  \cM_{\rho}(g(z)) = g(z) + \Delta\left[\lambda s_{11}\psi(z) +(1-2\lambda s_{21})\right].
\end{equation*}
If $g(z)$ is algebraic, then the function in the square bracket has to vanish identically. But 
\begin{equation*}
  \lambda s_{11}\psi(z) +(1-2\lambda s_{21}) = (1-2\lambda s_{21}) - \frac{2\lambda s_{11}}{\sqrt{z}}\left(1+\cdots\right).
\end{equation*}
Because $\lambda\neq0$, and $s_{11}$ does not vanish, we conclude that the
function in the square bracket cannot vanish identically. Therefore, $g(z)$ is
not algebraic. This finishes the proof of Lemma~\ref{lem:necsu}. 

\section{Proof of Lemma~\ref{lem:necsu_m}}
\label{sec:proof_lemma_m} 
In this section, we consider system~\eqref{eq:reduced} with $r(z)$ and $s(z)$ given by \eqref{eq:rr_ss_k=2_m}. Recall that for the considered case
the integrals $\psi(z)$, $\varphi(z)$, and $I(z)$ are given
by~\eqref{eq:psi_k=2_m} and~\eqref{eq:phi_k=2_m}, that is 
\begin{equation}
\label{eq:psi_phi_I_k=2_m}	
\begin{split}
  \psi(z)=&\int \frac{1}{x_1(z)^2}\rmd z = -\frac{2}{\sqrt{z}} \, 
F\left(-\frac{1}{2},1+\frac{1}{m};\frac{1}{2};z\right), \\
\varphi(z)=&-\frac{m \Omega  (1-z)^{2/m} (m+2 z)}{2 (m+2)},\\
I(z)=&\int \varphi(z)\psi(z)\rmd z
= a(z)+b\sqrt{z}F\left(\frac{1}{2},1 -\frac{1}{m};\frac{3}{2};z\right),
\end{split}
\end{equation}
where $a(z)$ is an algebraic function and $b$ is a non-zero constant (their
explicit forms are irrelevant for our further considerations); function $x_1(z)$
given by~\eqref{eq:X_sol_1_k=2_m} is an algebraic solution of
equation~\eqref{eq:reduced_a}. Because function $\varphi(z)$ is algebraic, we
will use criterion given in Lemma~\ref{lem:necsu2}. For the given forms of
integrals $\psi(z)$ and $I(z)$, it is clear that, after the rearrangement of
terms, as a function $g(z)$ in Lemma~\ref{lem:necsu2} we can take 
\begin{equation}
	\label{eq:is alg}
	g(z):= 
		\lambda F(\alpha,\beta;\gamma;z)+
  zF(\widehat\alpha,\widehat\beta;\widehat\gamma;z),
\end{equation} 
where  
\begin{equation}
	\label{eq:para}
	\begin{split}
		(\alpha,\beta,\gamma) = &\left(-\frac{1}{2},1+\frac{1}{m},\frac{1}{2}\right),\quad
		(\widehat\alpha,\widehat\beta,\widehat\gamma) = 
		\left(\frac{1}{2},1-\frac{1}{m},\frac{3}{2}\right).
	\end{split}
\end{equation}
From Lemma~\ref{lem:necsu0} we know that if integral $\psi(z)$ is algebraic,
then  the identity component of the differential Galois group of
system~\eqref{eq:reduced} is Abelian. Thus, we have to check if it is possible. 
\begin{proposition}
\label{lem:F_m}
If $\abs{m}>2$, then function $f(z)= F(\alpha,\beta;\gamma;z)$ with parameters $(\alpha,\beta,\gamma)$ given by~\eqref{eq:para} is not algebraic.
\end{proposition}
\begin{proof}
The proof is similar to that of Proposition~\ref{prop:reduced}.  The Gauss
hypergeometric equation\eqref{eq:33} with parameters $(\alpha,\beta,\gamma)$ is
reducible and it has one algebraic solution $w_1(z)=\sqrt{z}$. If it second
solution $w_2(z)=f(z)$  is linearly independent from $w_1(z)$ and is algebraic
then it it has the form
\begin{equation}
    \label{eq:}
    f(z)= z^{e_0}(z-1)^{e_1}p(z)
  \end{equation}
where $p(z)$ is a polynomial of degree $d=-(e_0+e_1+e_\infty)$, $e_0$, $e_1$ and $e_\infty$ are exponents of the hypergeometric equation. 
Thus
  \begin{equation*}
    \label{eq:e_0_e_1_m}
    e_0 = \left\{0,\frac{1}{2}\right\}, \quad e_1 = \left\{0,-\frac{1}{m}\right\}, \quad
    e_\infty = \left\{-\frac{1}{2},1+\frac{1}{m}\right\}.
  \end{equation*} 
By assumption $\abs{m}>2$. Hence, we have only one possibility:
$e_0=\tfrac{1}{2}$, $e_1=0$, and $e_\infty=-\frac{1}{2}$, $d=0$.  But with this
choice we get $f(z)=c\sqrt{z}=cw_1(z)$. However, by assumption $f(z)$  and
$w_1(z)$ are linearly independent. This contradiction shows that $f(z)$ is not
algebraic. 
\end{proof}
Similarly we prove the following proposition.
\begin{proposition}
\label{lem:F_m_hat}
If $\abs{m}>2$, then function $\widehat{f}(z)= F(\widehat\alpha,\widehat\beta;\widehat\gamma;z)$ with parameters $(\widehat\alpha,\widehat\beta,\widehat\gamma)$ given by~\eqref{eq:para} is not algebraic.
\end{proposition}
With these two propositions, we are ready to prove Lemma~\ref{lem:necsu_m}.
We will follow reasoning presented in Section~\ref{ssec:algebraic}. 
Thus, we have to investigate two  Gauss hypergeometric equations with respective  parameters
$(\alpha,\beta,\gamma)$ and
$(\widehat\alpha,\widehat\beta,\widehat\gamma)$. The respective bases of
local solutions  in a neighbourhood of  singularity $z=0$ are 
\begin{align*}
  u_1(z) = & F(\alpha,\beta;\gamma;z),                               \qquad\qquad
  u_2(z)= \sqrt{z},  \\ \widehat{u}_1(z) =
           & F(\widehat\alpha,\widehat\beta;\widehat\gamma;z), \qquad\qquad
  \widehat{u}_2(z)=  \frac{1}{\sqrt{z}}. 
\end{align*}
Note, that we have the we have similar situation as in Section~\ref{ssec:algebraic}. The only difference is that other parameters of hypergeometric functions. 

Next, in the same way as in Section~\ref{ssec:algebraic}, we construct the
global monodromy matrices $M_{\sigma_1}$ and $\widehat M_{\sigma_1}$
corresponding to the loop $\sigma_1$ around singularity $z=1$. Then, we compute
the respective commutator matrices $C$ and $\widehat C$, and $D$ and $\widehat
D$. Finally, we compute the determinant $\Delta$ given by~\eqref{eq:20}. We obtain
\begin{equation*}
  \Delta = 4 \rmi\pi  (m+2) \sin ^2\left(\frac{2 \pi }{m}\right).
\end{equation*}
Because $\abs{m}>2$, we have $\Delta\neq 0$. Thus, function $g(z)$
given by~\eqref{eq:is alg} is not algebraic. Therefore, according to
Lemma~\ref{lem:necsu0}, the identity component of the differential Galois group of
 system~\eqref{eq:reduced} is not Abelian. This completes the proof of Lemma~\ref{lem:necsu_m}.
\section{Applications of the integrability obstructions \label{sec:app}}
This section presents the application of the obtained integrability obstructions to the Hamilton equations of motion~\eqref{eq:Ham-eq} governed by Hamiltonian~\eqref{eq:Hamek}, as formulated in Theorems~\ref{th:main_theorem}--\ref{th:main_theorem_k=2}. 
We demonstrate their effectiveness and simplicity of use by performing only basic algebraic computations to establish the non-integrability of the considered Hamiltonians and to identify parameter values for which integrability may still be possible.

In addition, to gain qualitative insight into the dynamics of the studied systems, we complement the analytical approach with a numerical analysis based on Poincar\'e cross-sections. 
This analysis illustrates how variations in the system parameters influence the overall dynamics and integrability, typically leading to the onset of chaotic behavior.

\subsection{The generalized Hill model}
As the first example, let us consider a generalized version of the planar circular Hill problem.  
In its classical form, the Hill problem arises as a limiting case of the restricted three-body problem, 
describing the motion of a massless body in the vicinity of a smaller primary under the gravitational influence of a massive one.  
The model was originally introduced by George W.~Hill~\cite{Hill_1905} in his study of the lunar motion within the Sun--Earth--Moon system, 
and it has since become a cornerstone in celestial mechanics for analyzing local dynamics in rotating frames.  

Over time, several modifications and extensions of the Hill model have been proposed.  
These include~relativistic~and~pseudo-Newtonian generalizations~\cite{Steklain_2006,Steklain_2009,Zotos_2019},
photo-gravitational versions that account for radiation pressure~\cite{Kanavos_2002, Markellos_2000}
and Hill-type approximations applied to stellar dynamics and galactic motion~\cite{Heggie_2001}.  
Extensive numerical investigations have also been performed, 
revealing complex families of periodic, escaping, and chaotic orbits that emerge in different parameter regimes~\cite{Deng_2021,Combot_2022}.

Following these developments, we consider here a generalized Hill system, which extends the classical model by including additional quadratic terms in the potential energy.  
The dynamics is completely governed by the Hamiltonian
\begin{equation}
\label{eq:Hill_gen}
 	H=\frac{1}{2}\left(p_1^2+p_2^2\right)
	+\omega(q_2p_1-q_1p_2)
	-\frac{\mu}{r}
	+ A\,q_1^2+B\,q_2^2,
\end{equation}
where $\mu, A, B \in \mathbb{R}^+$ and $r^2 = q_1^2 + q_2^2$.  
Here, the parameters $A$ and $B$ introduce anisotropy into the gravitational field, 
representing the effects of tidal and rotational deformations of the potential.  
The gyroscopic term $\omega(q_2p_1 - q_1p_2)$ accounts for the Coriolis and centrifugal forces in the rotating reference frame, 
while the central term $-\mu/r$ models the gravitational attraction of the dominant primary body.  

This generalized form of the Hill Hamiltonian bridges the gap between the classical lunar Hill problem 
and modern galactic or stellar-dynamical models that include rotation and anisotropic perturbations.  

In the recent paper~\cite{2023NonDy.111..275P}, the authors proved the non-integrability of the generalized Hill model for the parameter values $A=-1$ and $B=\tfrac{1}{2}$. 
They applied the differential Galois approach together with the Kovacic algorithm in dimension four, recently formulated in~\cite{Combot:18b::}, to analyse the structure of the differential Galois group of the variational equations, which in this case are of the fourth order. 
Through an advanced and highly technical computation, they proved that the generalized Hill system is not Liouville integrable in this configuration. 

The motivation of the present section is to show how the non-integrability of this model can be established in a straightforward and transparent way by applying the analytical criterion formulated in this paper. 
Unlike the earlier work, our approach avoids the complicated analysis of higher-order differential equations and relies solely on simple algebraic operations.

We state the following proposition.
\begin{proposition}\label{prop:hill}
For $\omega\,\mu\neq 0$ and $A\neq B$, the generalized Hill model governed by
 Hamiltonian~\eqref{eq:Hill_gen} is not integrable in the Liouville sense with meromorphic first integrals. 
\end{proposition}
\begin{proof}
We can rewrite the potential in Hamiltonian~\eqref{eq:Hill_gen} in the general form~\eqref{eq:VV} by identifying a single homogeneous component of degree $k=2$ and treating the second component as vanishing, i.e.
\[
V_k = A\, q_1^2 + B\, q_2^2,\qquad V_m=0.
\]
According to the definitions~\eqref{eq:n,l}--\eqref{eq:coeffs}, the corresponding parameters and integrability coefficients are
\[
k=2,\quad 
\lambda_k = V_k(1,\rmi) = A-B,\quad 
\lambda_m = V_m(1,\rmi) = 0.
\]
Thus, whenever $A\neq B$, the system satisfies $\lambda_k\neq 0$ and $\lambda_m=0$.  
Under these conditions, the assumptions of Theorem~\ref{th:main_theorem_k=2} are fulfilled.  
Consequently, the Hamiltonian~\eqref{eq:Hill_gen} admits no meromorphic first integral functionally independent of $H$, 
and the system is therefore not Liouville integrable for any anisotropic configuration $A\neq B$. \end{proof}
It is worth noting that the parameter values investigated in earlier studies (for instance, $A=-1$ and $B=\tfrac{1}{2}$) fall precisely within this anisotropic regime.  
Hence, the non-integrability of the Hill problem established here using our analytical integrability criterion is fully consistent with the results previously obtained through the analysis of fourth-order variational equations.  
In contrast to that highly technical approach, our proof relies solely on simple algebraic computations and a direct verification of the necessary conditions formulated in Theorem~\ref{th:main_theorem_k=2}, providing a much more straightforward and transparent demonstration of the system’s non-integrability.

The only integrable configuration corresponds to the radial case, when $\lambda_k=\lambda_m=0$, which implies $A=B$.  
This represents a degenerate isotropic situation in which the Hamiltonian possesses full rotational invariance.  
In this case, the system admits an additional first integral associated with the conservation of angular momentum, and therefore becomes Liouville integrable.

\subsection{Anisotropic polynomial potential}
 	\begin{figure*}[htp]
		\centering
		\subfigure[$A=B=1$ ]{
		\includegraphics[width=0.4\linewidth]{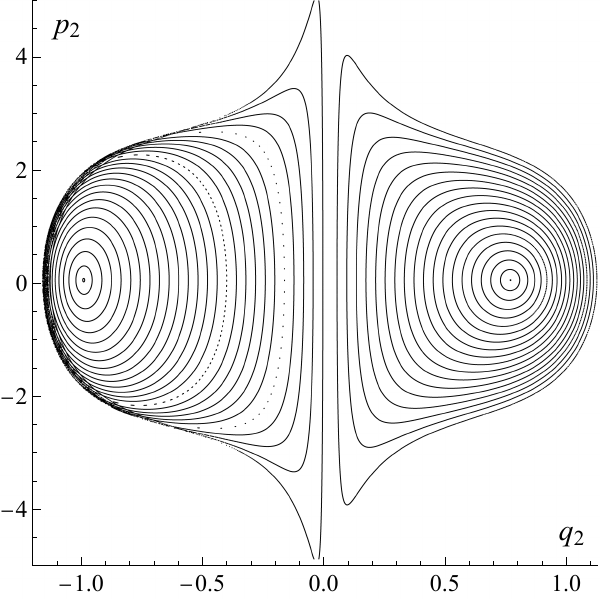}}\hspace{1cm}
		\subfigure[$A=1,B=\tfrac{9}{10}$ ]{
			\includegraphics[width=0.4\linewidth]{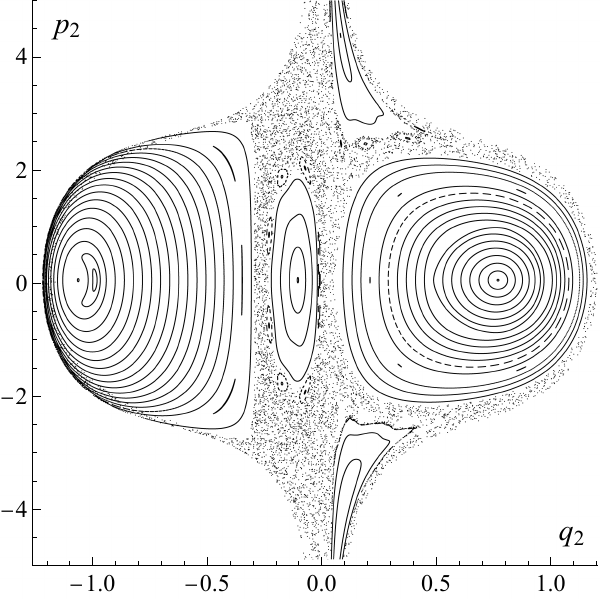}}
			\caption{The Poincar\'e sections of  system~\eqref{eq:Hamek} with  potential~\eqref{eq:v_radial} computed for $\omega= \mu=\tfrac{1}{10}$ and $m=4$ at constant energy level $E=2$. The cross-section plane was specified as $q_1=0$, and the direction was chosen by $p_1>0$.\label{fig:radial0}}
			 	\end{figure*}
					 	\begin{figure*}[htp]
		\centering
		\subfigure[$A=B=1$ ]{
		\includegraphics[width=0.4\linewidth]{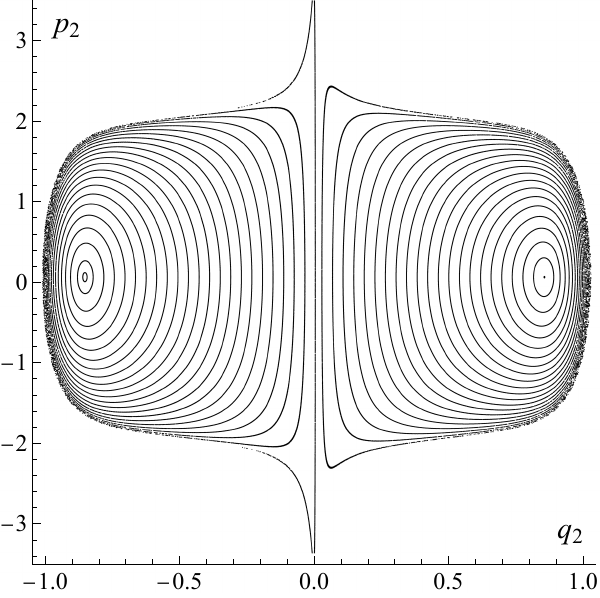}}\hspace{1cm}
		\subfigure[$A=1,B=\tfrac{9}{10}$ ]{
			\includegraphics[width=0.4\linewidth]{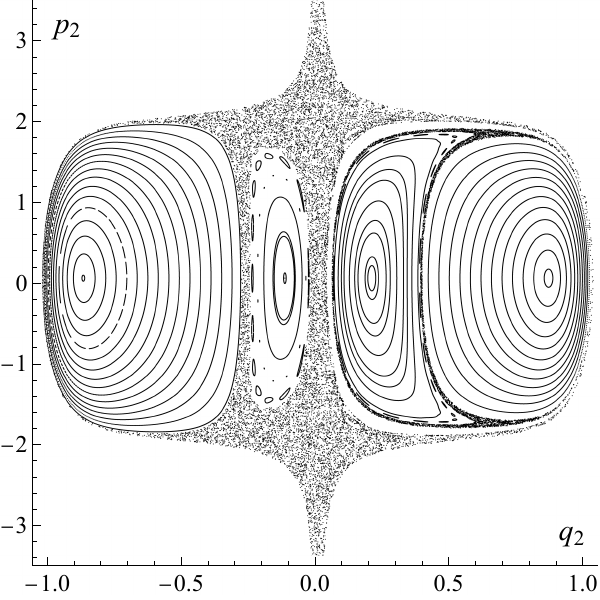}}
			\caption{The Poincar\'e sections of  system~\eqref{eq:Ham-eq} with   potential~\eqref{eq:v_radial} computed for $\omega=\mu=\tfrac{1}{10}$ and $m=6$ at constant energy level $E=2$. The cross-section plane was specified as $q_1=0$, and the direction was chosen by $p_1>0$.\label{fig:radial}}
	\end{figure*}
Let us now consider the Hamiltonian~\eqref{eq:Hamek} with the potential
\begin{equation}
\label{eq:v_radial}
V(q_1,q_2)=\frac{1}{2}\!\left(A\, q_1^2+B\, q_2^2\right)+(q_1^2+q_2^2)^{m/2},
\end{equation}
where $A,B$ are real parameters and $m \in \Q$. 

This potential consists of two physically distinct contributions.
The quadratic part, $\tfrac{1}{2}(A\, q_1^2 + B\, q_2^2)$, describes a two-dimensional anisotropic harmonic oscillator, while the nonlinear radial term, $(q_1^2+q_2^2)^{m/2}$, introduces an isotropic  coupling depending only on the distance from the origin. 

To gain insight into the dynamics of the considered model, we constructed a pair of Poincar\'e cross-sections for two representative values of the exponent, $m=4$ and $m=6$, shown in Figs.~\ref{fig:radial0}--\ref{fig:radial}.
They clearly illustrate the qualitative transition from integrable to non-integrable behavior as the system parameters are varied. 
For $A = B = 1$, the motion is entirely regular: all trajectories lie on smooth, closed invariant curves corresponding to quasi-periodic motion on invariant tori. 
Families of quasi-periodic trajectories enclose two periodic orbits, and no signatures of chaos are observed. 
The system remains perfectly symmetric with respect to the potential wells. 

When the symmetry is slightly broken, for example for $A = 1$ and $B =\tfrac{9}{10}$, the integrability is lost. 
Although some invariant tori persist, forming regular islands around stable periodic orbits, large regions of the phase space become chaotic. 
Trajectories in these areas exhibit irregular scattering, characteristic of deterministic chaos. 
The resulting mixed phase-space structure, where regular and chaotic domains coexist, follows the Kolmogorov–Arnold–Moser (KAM) scenario, demonstrating the gradual destruction of invariant tori under small perturbations.

As shown in Figs.~\ref{fig:radial0}--\ref{fig:radial}, the system is generally non-integrable when $A \neq B$. 
To verify whether this loss of integrability persists for all values of the exponent $m$, we now apply our analytical integrability obstructions. We can prove the following proposition.

\begin{proposition}
\label{th:nonint_quadratic_plus_radial}
Consider the rotating Hamiltonian system~\eqref{eq:Hamek} with the potential~\eqref{eq:v_radial},
so that $k=2$, with real parameters $A,B$ and rational $m\in\mathbb{Q}$.  
Assume that $\mu\,\omega\neq 0$.  
Then, for $A\neq B$, the system is \emph{not} Liouville integrable with meromorphic first integrals for any rational exponent~$m\in \Q$.
\end{proposition}

\begin{proof}
By identifying the potential~\eqref{eq:v_radial} with the general form~\eqref{eq:VV}, we have
\[
k=2,\quad 
V_k=\frac{1}{2}(A\,q_1^2+B\,q_2^2),\quad 
V_m=(q_1^2+q_2^2)^{m/2} .
\]
Using definitions~\eqref{eq:n,l} and~\eqref{eq:coeffs}, we obtain \[
n=0,\quad 
l=\frac{8}{m-2},\quad 
\lambda_k=\frac{A-B}{2},\quad 
\lambda_m=0.
\]

Hence, for $A\neq B$ we have $\lambda_k\neq 0$ and $\lambda_m=0$.  
The assumptions of Theorem~\ref{th:main_theorem_k=2} are therefore satisfied with $k=2$.  
According to this theorem, the system admits no meromorphic first integral functionally independent of the Hamiltonian.  
Consequently, the Hamiltonian~\eqref{eq:Hamek} with the potential~\eqref{eq:v_radial} is not Liouville integrable for any rational exponent~$m\in \Q$. 
\end{proof}

The only integrable case corresponds to the radial case, when $\lambda_k=\lambda_m=0$, which implies $A=B$.  
This represents a degenerate isotropic situation in which the Hamiltonian possesses rotational invariance.  
We can state the following straightforward observation.

\begin{proposition}\label{prop:calka}
For $\omega,\mu\in\C$ and $k,m\in\Q$, the Hamiltonian system governed by
\begin{equation}
\label{eq:radial}
\begin{split}
H&=\frac{1}{2}\left(p_1^2+p_2^2\right)
+\omega\,(q_2p_1-q_1p_2)
-\frac{\mu}{\sqrt{q_1^2+q_2^2}}\\
&\quad
+\frac{1}{2}\!\left(q_1^2+q_2^2\right)
+(q_1^2+q_2^2)^m,
\end{split}
\end{equation}
is Liouville integrable with the additional first integral
\begin{equation}
\label{eq:ii}
I = q_2 p_1 - q_1 p_2.
\end{equation}
\end{proposition}

From a geometric perspective, this result confirms that any deviation from axial symmetry immediately breaks angular momentum conservation and leads to the onset of non-integrable dynamics, in full agreement with the numerical evidence from the Poincar\'e sections.

\subsection{The  H\'enon-Heiles galactic potential}
		 	\begin{figure*}[htp]
		\centering
		\subfigure[$a=0,\, b=\tfrac{1}{2},\, A=1,\, B=1,\, E=\tfrac{3}{5}$  ]{
		\includegraphics[width=0.4\linewidth]{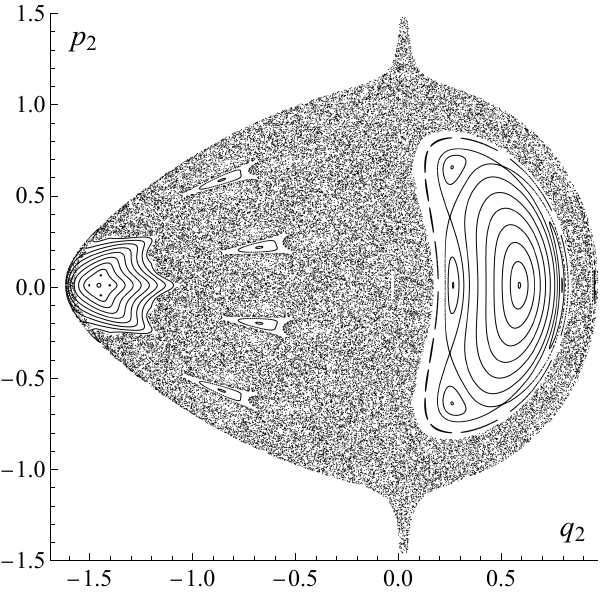}}\hspace{1cm}
		\subfigure[$a=\tfrac{1}{2},\, b=\tfrac{1}{2},\, A=1,\, B=1, \,E=\tfrac{3}{100}$  ]{
			\includegraphics[width=0.4\linewidth]{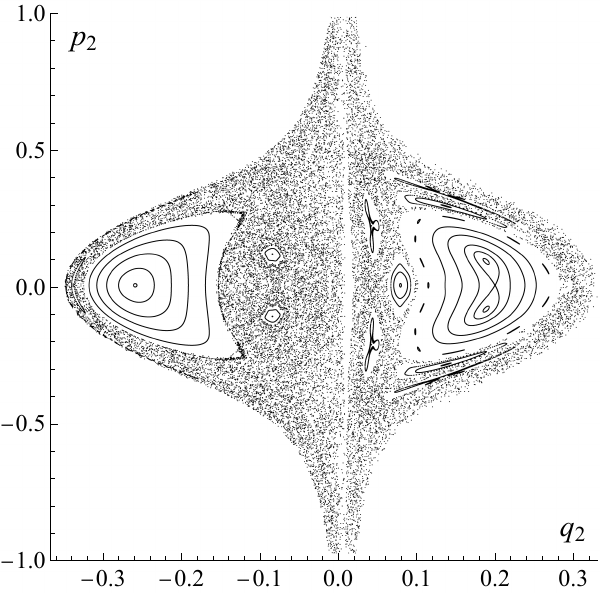}}\\ \vspace{10pt}
					\subfigure[$a=\tfrac{1}{12},\, b=\tfrac{1}{2},\, A=1,\, B=1,\, E=\tfrac{1}{2}$  ]{
			\includegraphics[width=0.4\linewidth]{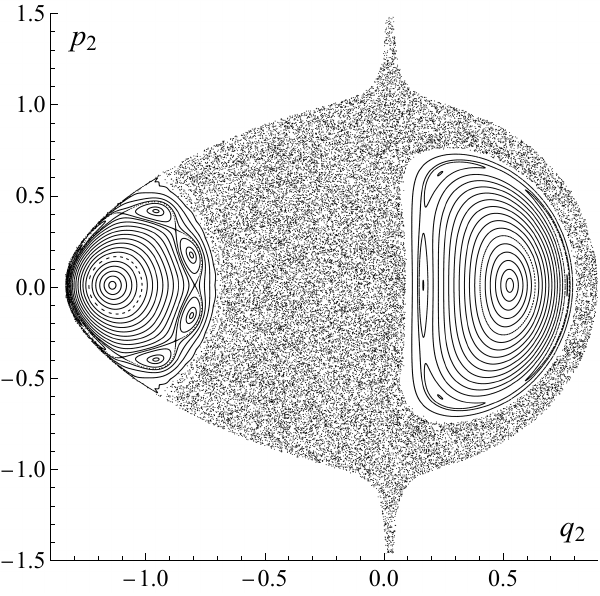}}\hspace{1cm}
			\subfigure[$a=1,\, b=16,\, A=1,\, B=16, \,E=1$  ]{
			\includegraphics[width=0.4\linewidth]{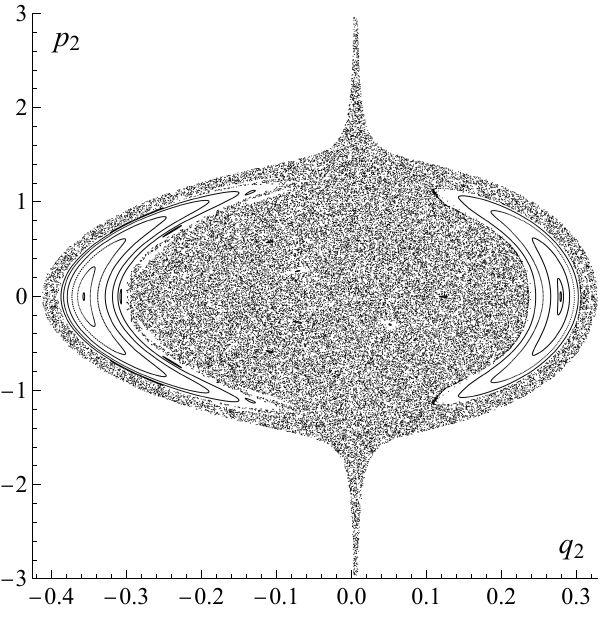}}
			\caption{The Poincar\'e sections of system~\eqref{eq:Ham-eq} with the H\'enon--Heiles potential~\eqref{eq:henon} were computed for 
$\omega=\tfrac{1}{10}$ and $\mu=\tfrac{1}{100}$, with varying parameters $A,B,a,b$ at constant energy levels~$E$.  
The cross-section plane was defined by $q_1=0$, and the direction was chosen according to $p_1>0$.  
The parameter values correspond to the integrable cases of the classical H\'enon--Heiles model given by~\eqref{eq:henon_par}, 
that is, for $\omega=\mu=0$.  
As can be observed, nonzero values of $\omega$ and~$\mu$ destroy the system’s integrability.  
The resulting figures indicate non-integrability through the emergence of chaotic behavior.	\label{fig:henon}}
	\end{figure*}
	
The classical H\'{e}non--Heiles potential constitutes one of the simplest and most celebrated examples of nonlinear Hamiltonian systems that can exhibit both regular and chaotic behaviour, depending on the choice of parameters and the total energy of the system.  
In its general form,  is defined by the following non-homogeneous potential 
\begin{equation}
\label{eq:henon}
V(q_1,q_2)=\frac{1}{2}\left(
	A\, q_1^2+ B\, q_2^2\right)+a\, q_1^2q_2 +\frac{b}{3}q_2^3,
\end{equation}
where $A,B,a,b$ are real parameters.  

Originally introduced in the context of stellar motion in an axisymmetric galactic potential~\cite{Henon:64::}, the H\'{e}non--Heiles model has since become a paradigmatic system in the study of deterministic chaos. It has found numerous applications in different areas of physics, ranging from celestial mechanics and nonlinear oscillations to statistical and quantum mechanics~\cite{Ford:73::}. More recently, it has even been employed as a benchmark model in the development and testing of Hamiltonian neural networks~\cite{Mattheakis:22::}, illustrating its enduring relevance across both theoretical and computational physics.  

The integrability of the H\'{e}non--Heiles Hamiltonian has been extensively studied. There exist precisely four known integrable cases~\cite{Chang:82::,Grammaticos:83::}, corresponding to the parameter ratios
\begin{equation}
\begin{split}
\label{eq:henon_par}
&\text{(1)}\quad a=0, \quad &&A,B,b\in \R,\\ 
&\text{(2)}\quad a=b,\quad &&A=B,\\ 
&\text{(3)}\quad b=6a,\quad  &&A,B\ \in \R,\\ 
&\text{(4)}\quad b=16a, \quad  &&B=16A.
\end{split}
\end{equation} 
for which the system admits an additional independent first integral. In the first three cases, the Hamiltonian becomes separable after an appropriate canonical transformation, which immediately yields an additional first integral. In contrast, in the fourth case the corresponding first integral is a quartic polynomial in the momenta, making this situation substantially more intricate from the dynamical and algebraic viewpoint. 

It was rigorously shown by Ito~\cite{Ito:85::}, and later confirmed and refined by Morales--Ramis theory~\cite{Li:11::}, that these are the only parameter values for which the classical H\'{e}non--Heiles system remains integrable.  
For all other choices of $a$ and $b$, the dynamics exhibit a rich mixture of regular and chaotic trajectories, providing one of the earliest and most iconic examples of the transition to chaos in Hamiltonian systems.  


Recently, increasing attention has been devoted to the study of generalised H\'enon--Heiles potentials, in which a rotational (gyroscopic) term is incorporated into the Hamiltonian~\eqref{eq:Hamek}. Such extended systems constitute a natural framework for analysing the influence of rotational effects on the integrability and global dynamics of non-homogeneous potentials.  
While the classical, non-rotating H\'enon--Heiles models have been extensively investigated over the past decades, the dynamical behaviour and integrability properties of their rotating analogues remain an active and demanding topic of research. Only partial results are currently available. For instance, in~\cite{Lanchares_2021}, the author examined the integrability of the Hamiltonian~\eqref{eq:Hamek} in the special case of a vanishing Kepler term ($\mu = 0$). In contrast, the studies~\cite{Zotos_2020_HHsingular,Navarro_2021_EPJPlus} focused on the combined H\'enon--Heiles and Kepler potentials, yet without accounting for the rotational contribution.  

To illustrate the dynamical consequences of addition of the gyroscopic and Keplerian terms, we analysed the system~\eqref{eq:Ham-eq} with the non-homogeneous H\'{e}non--Heiles potential~\eqref{eq:henon}.  
Figure~\ref{fig:henon} presents several Poincar\'e cross-sections computed for the parameter values~\eqref{eq:henon_par} corresponding to the integrable cases of the classical, non-rotating model.  
As can be clearly seen, the mentioned additional terms largely affect the system's dynamics --- the Poincar\'e sections reveal widespread chaotic regions interspersed with small islands of regular motion.  
Most of the accessible domains of the Poincar\'e planes are densely filled with scattered points, which correspond to non-integrable trajectories.  

Let us now study the integrability of  the Hamiltonian~\eqref{eq:Hamek} with the H\'enon--Heiles potential~\eqref{eq:henon}. The potential $V$ is a sum of two components
\begin{equation}
\label{eq:VkVM_henon}
V_k=\frac12\!\left(A\,q_1^2+B\,q_2^2\right),\qquad 
V_m=a\,q_1^2q_2+\frac{b}{3}\,q_2^3,
\end{equation}
for which $k=2$ and $m=3$.  
Using definitions~\eqref{eq:n,l} and~\eqref{eq:coeffs}, we obtain the corresponding parameters:
\[
n=0,\quad 
l=8,\quad 
\lambda_k=\frac{A-B}{2},\quad 
\lambda_m=\Bigl(a-\frac{b}{3}\Bigr)\rmi.
\]
Now, we state the following proposition.

\begin{proposition}\label{th:int_henon_necessary}
For the rotating Hamiltonian~\eqref{eq:Hamek} with the cubic (H\'enon--Heiles--type) potential~\eqref{eq:henon}, 
assume $\mu\,\omega\neq 0$. If the system is Liouville integrable with meromorphic first integrals, then  
\begin{equation}
\label{eq:warunki_henon}
\lambda_k=\lambda_m=0,
\qquad\text{i.e.}\qquad
A=B\ \ \text{and}\ \ b=3a.
\end{equation}
\end{proposition}

\begin{proof}
We distinguish three distinct regimes for the pair $(\lambda_k,\lambda_m)$.

\emph{Case (i):} $\lambda_k\,\lambda_m\neq 0$.  
Then Theorem~\ref{th:main_theorem} applies.  
For $n=0$ and $l=8$, Item~\ref{it:necs1} requires either $l\ge -1$ odd or $l<-1$ even; here $l=8$ is even and $l\ge -1$, so this fails.  
For Item~\ref{it:necs2}, $(n+l)(n+l+2)=8\cdot10\neq 0$, but $n=0$ is neither a positive even nor a negative odd integer; moreover, $n+l=8$ is even and falls into the excluded subcase $\{\,l\ge 0,\ n\le 0\ \text{both even}\,\}$.  
Hence integrability is impossible in this case.

\emph{Case (ii):} $\lambda_k=0$ and $\lambda_m\neq 0$.  
With $k=2$ and $m=3\in\mathbb{Q}$ (thus $|m|>2$), Theorem~\ref{th:main_theorem_k=2_m} yields non-integrability.

\emph{Case (iii):} $\lambda_k\neq 0$ and $\lambda_m=0$.  
With $k=2$ and $m\in\mathbb{Q}$, Theorem~\ref{th:main_theorem_k=2} implies there is no meromorphic first integral functionally independent of the Hamiltonian.

Since in each of the Cases (i)–(iii) integrability is excluded, the only remaining possibility is
$\lambda_k=\lambda_m=0$, i.e., $A=B$ and $b=3a$.  
This proves the stated necessary condition. 
\end{proof}

The above proposition establishes the necessary conditions under which the rotating Hamiltonian~\eqref{eq:Hamek} with the H\'enon--Heiles potential~\eqref{eq:henon} could be Liouville integrable.  
In particular, integrability is possible only in the degenerate configuration $\lambda_k=\lambda_m=0$, corresponding to $A=B$ and $b=3a$.  
All other parameter combinations violate these necessary conditions and therefore lead to non-integrability.  
Consequently, the parameter values~\eqref{eq:henon_par} that correspond to the integrable cases of the classical H\'enon--Heiles model no longer yield integrability once the gyroscopic term $(\omega\neq 0)$ and the Kepler potential term $(\mu\neq 0)$ are included in the Hamiltonian.

The remaining degenerate configuration $A=B$ and $b=3a$ formally satisfies the necessary conditions for integrability but does not guarantee it.  
As illustrated in Fig.~\ref{fig:henon2}, the Poincar\'e section computed for this parameter set at a fixed energy level $E=-\tfrac{1}{100}$ reveals a mixed phase--space structure, where regular invariant curves coexist with scattered points corresponding to chaotic trajectories.  
The progressive destruction of invariant tori and the emergence of stochastic layers near the separatrix clearly indicate the breakdown of integrability and the onset of chaotic dynamics.  
A rigorous proof of non-integrability in this special degenerate case would require an analysis of the higher-order variational equations.
		 	\begin{figure}[t]
		\includegraphics[width=0.8\linewidth]{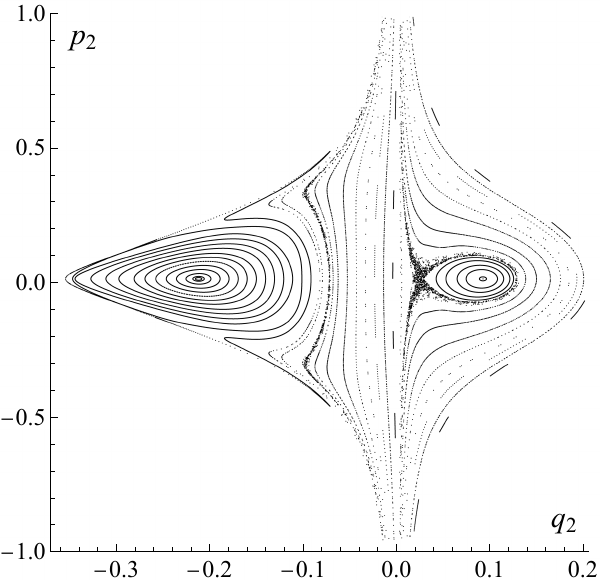}
\caption{The Poincar\'e sections of system~\eqref{eq:Ham-eq} with the H\'enon--Heiles potential~\eqref{eq:henon} were computed for 
$\omega=\tfrac{1}{10}$ and $\mu=\tfrac{1}{100}$, with $A=B=1$ and $a=1,\,b=3$, at the constant energy level~$E=-\tfrac{1}{100}$.  
The cross-section plane is defined by $q_1=0$, with the direction specified by $p_1>0$.  
Although the necessary conditions for integrability are satisfied, the figure clearly demonstrates the system’s non-integrability through the onset of chaotic motion in the vicinity of the separatrix.
\label{fig:henon2}}
	\end{figure}
	
\subsection{The generalized quartic galactic potential}
As a last but not least important example,
let us now examine the integrability of a system governed by a Hamiltonian~\eqref{eq:Hamek} with a generalized quartic galactic potential of the form
\begin{equation}
\label{eq:quartic}
  V(q_{1},q_{2})=\frac{1}{2}\!\left(A\,q_1^2+B\,q_2^2\right)
  +a\,q_1^4+b\,q_1^2 q_2^2+c\,q_2^4,
\end{equation}
where $A,B,a,b,c\in\R$ are real parameters.  

This general form encompasses a broad class of astrophysical and dynamical systems of interest, 
providing a unified theoretical framework for analysing a variety of models. 
Several important special cases can be recovered within this formulation. 
For $a = b$ and $c = a + b$, the potential~\eqref{eq:quartic} reduces to the classical Armbruster--Guckenheimer--Kim (AGK) galactic potential~\cite{Armbruster_1989,Llibre_2019}, which describes the motion of stars in a triaxial galactic field. 
For $A = 1$ and $B = 1/k$, where $k \in (0,1]$, the potential~\eqref{eq:quartic} represents a galactic model describing flattened or elliptical galaxies~\cite{Lacomba_2012}. 
Finally, when $A = B = 1$ and $c = 0$, the potential~\eqref{eq:quartic} corresponds to a generalized Yang--Mills--typepotential~\cite{Kasperczuk_1994,JimenezLara_2011_JMP}, 
which naturally appears in certain field-theoretic and nonlinear oscillator contexts.

From a physical point of view, such quartic potentials describe systems with both harmonic and anharmonic restoring forces,  
where the quadratic terms model small oscillations near equilibrium, and the quartic terms represent nonlinear corrections responsible for coupling and energy exchange between modes.  
In galactic dynamics, they are used to approximate the motion of stars in non-axisymmetric gravitational fields,  
where deviations from a purely harmonic potential account for the observed chaotic structures in stellar orbits.  
Hence, the study of integrability for the Hamiltonian~\eqref{eq:Hamek}  with potential~\eqref{eq:quartic} provides valuable insight into the transition between regular and chaotic motion in realistic galactic models.

We shall rigorously demonstrate that the rotating Hamiltonian~\eqref{eq:Hamek} with the quartic galactic potential is non-integrable for all complex parameter values $A,B,a,b,c\in\mathbb{C}$ satisfying $\mu\,\omega\neq 0$.  

Potential~\eqref{eq:quartic} is the sum of two polynomial terms
\begin{equation}
\label{eq:quartic_VkVM}
V_k=\frac12\!\left(A\,q_1^2+B\,q_2^2\right),\quad 
V_m=a\,q_1^4+b\, q_1^2 q_2^2+c\,q_2^4,
\end{equation}
for which $k=2$ and $m=4$.  
Using definitions~\eqref{eq:n,l} and~\eqref{eq:coeffs}, we obtain the corresponding parameters:
\[
n=0,\quad 
l=4,\quad 
\lambda_k=\frac{A-B}{2},\quad 
\lambda_m=a-b+c.
\]
Now, we state the necessary condition for possible integrability.

\begin{proposition} 
\label{th:int_quartic_necessary}
For the Hamiltonian~\eqref{eq:Hamek} with the quartic potential~\eqref{eq:quartic_VkVM}, 
assume $\mu\,\omega\neq 0$. 
If the system is Liouville integrable with meromorphic first integrals, then necessarily
\begin{equation}
\label{eq:warunki_quartic}
\lambda_k=\lambda_m=0,
\quad\text{i.e.}\quad
A=B\ \ \text{and}\ \ a-b+c=0.
\end{equation}
\end{proposition}

\begin{proof}
We analyse all possible combinations of the coefficients $(\lambda_k,\lambda_m)$.

\emph{Case (i):} $\lambda_k\,\lambda_m\neq 0$.  
Then the assumptions of Theorem~\ref{th:main_theorem} apply.  
For $n=0$ and $l=4$, Item~\ref{it:necs1} requires either $l\ge -1$ odd or $l<-1$ even; 
here $l=4$ is even and $l\ge -1$, so this condition fails.  
For Item~\ref{it:necs2}, $(n+l)(n+l+2)=4\cdot6\neq 0$, 
but $n=0$ is neither a positive even nor a negative odd integer; 
moreover, $n+l=4$ is even and falls into the excluded subcase $\{\,l\ge 0,\ n\le 0\ \text{both even}\,\}$.  
Therefore, the necessary conditions of Theorem~\ref{th:main_theorem} are not satisfied, and the system cannot be integrable in this case.

\emph{Case (ii):} $\lambda_k=0$ and $\lambda_m\neq 0$.  
Here $A=B$, while $a-b+c\neq 0$. 
With $k=2$ and $m=4\in\mathbb{Q}$ (thus $|m|>2$), Theorem~\ref{th:main_theorem_k=2_m} 
directly implies non-integrability.

\emph{Case (iii):} $\lambda_k\neq 0$ and $\lambda_m=0$.  
Here $A\neq B$ and $a-b+c=0$. 
With $k=2$ and $m\in\mathbb{Q}$, Theorem~\ref{th:main_theorem_k=2} ensures that no meromorphic  first integral functionally independent of the Hamiltonian exists.

Since integrability is excluded in all the above configurations, the only remaining possibility for which the system might satisfy the necessary conditions for Liouville integrability is formulated in~\eqref{eq:warunki_quartic}.

\end{proof}
For the remaining values of the parameters~\eqref{eq:warunki_quartic} for which the system formally satisfies the necessary integrability conditions, the dynamics still appears to be non-integrable, as illustrated in Fig.~\ref{fig:gal} showing representative Poincar\'e sections.  
It seems that the only genuinely integrable case occurs when, in addition to condition~\eqref{eq:warunki_quartic}, the relation $a=c$ is also fulfilled.  
In this configuration, the system becomes truly integrable and belongs to the integrable family described earlier in Proposition~\ref{prop:calka}.

		 	\begin{figure*}[t]
			\subfigure[$A=1, B=1, a=1, b=2, c=1$]{ 
		\includegraphics[width=0.32\linewidth]{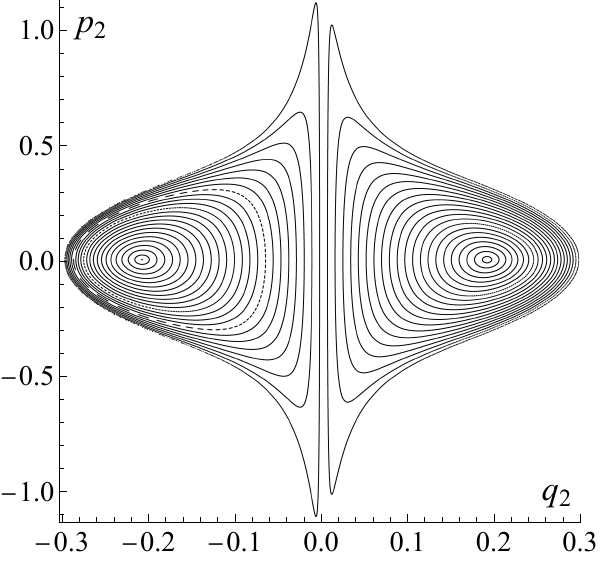}}
			\subfigure[$A=1, B=1, a=1, b=6, c=5$]{ 
			\includegraphics[width=0.32\linewidth]{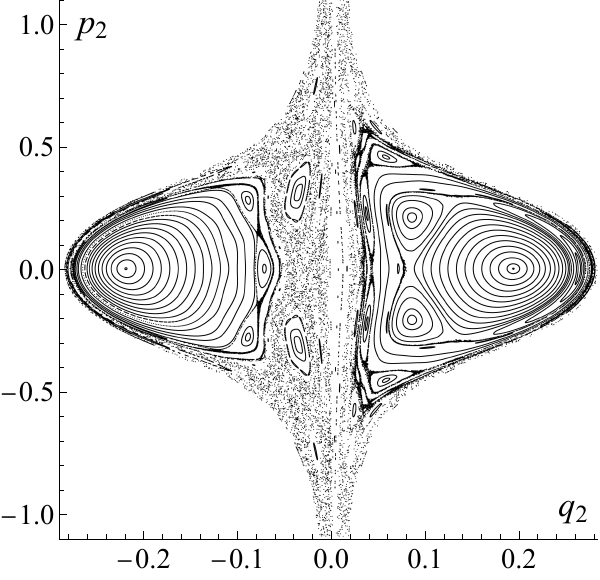}}
						\subfigure[$A=1, B=1, a=1, b=11, c=10$]{ 
				\includegraphics[width=0.32\linewidth]{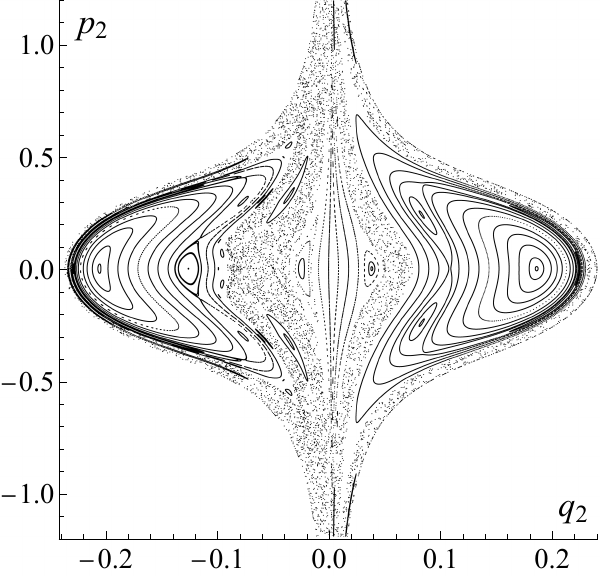}}
\caption{The Poincar\'e sections of system~\eqref{eq:Ham-eq} with the quartic (galactic) potential~\eqref{eq:quartic} were computed for 
$\omega=\tfrac{1}{10}$ and $\mu=\tfrac{1}{100}$ 
at the constant energy level~$E=\tfrac{1}{50}$. 
The remaining parameters $A,B$ and $a,b,c$ were chosen so as to satisfy the necessary integrability conditions formulated in~\eqref{eq:warunki_quartic}.  
The cross-section plane was defined by $q_1=0$, with the direction specified by $p_1>0$.  
Although the necessary integrability conditions are formally satisfied, the figure clearly shows that the system is integrable only for $a=c$, 
while for $a\neq c$ the Poincar\'e sections exhibit the onset of chaos, indicating the loss of integrability.}
\label{fig:gal}
	\end{figure*}

\section{Case $\mu=0$: the exceptional potentials \label{sec:exc}}
Consider the Hamiltonian~\eqref{eq:Hamek} in the absence of the Keplerian term, 
i.e., for $\mu=0$, so that it reduces to~$H_0$.  
In this case, the non-integrability of the regularised Hamiltonian~\eqref{eq:Ku} 
does not necessarily imply the non-integrability of the original Hamiltonian~\eqref{eq:Hamek}.  
Consequently, the integrability obstructions formulated in 
Theorems~\ref{th:main_theorem}--\ref{th:main_theorem_k=2} 
are not applicable here.  

Moreover, in the absence of an appropriate particular solution,
a complete integrability analysis of the Hamiltonian~$H_0$ within the differential Galois framework could not be carried out when the Keplerian term is removed.
Nevertheless, we have obtained an interesting and somewhat unexpected result:
in this setting, the rotating Hamiltonian~$H_0$ with a non-homogeneous exceptional potential turns out to be  super-integrable.

Specifically, we study the rotating Hamiltonian~$H_0$ with the 
non-homogeneous exceptional potential
\begin{equation}
\label{eq:exceptional}
V(q_1,q_2)=\left(q_2-\rmi q_1\right)^k+\left(q_2-\rmi q_1\right)^m,
\end{equation}
see~\cite{Maciejewski_2005_JMP,2023NonDy.111..275P} for the analysis of the homogeneous case.

Introducing the canonical variables 
$(x_1, x_2, y_1, y_2)$ defined by
\begin{align*}
x_1 &= q_2 - \rmi q_1, 
&\quad x_2 &= q_2 + \rmi q_1, \\[0.5em]
y_1 &= \frac{p_2 + \rmi p_1}{2}, 
&\quad y_2 &= \frac{p_2 - \rmi p_1}{2},
\end{align*}
we transform $H_0$ to the simpler form
\begin{equation}
\label{eq:ham_ex}
H_0(\vx,\vy)=2y_1y_2+\widetilde{\omega}\left(x_2y_2-x_1y_1\right)+x_1^k+x_1^m,
\end{equation}
where $\widetilde{\omega}=\rmi\omega$.  
Now we prove the following proposition. 

\begin{proposition}
\label{prop:superint}
Assume that $\widetilde{\omega}\in\C$ and $k,m\in\Q\setminus\{-1\}$.  
Then the Hamiltonian system governed by~\eqref{eq:ham_ex} is super-integrable.  
The corresponding first integrals take the following explicit forms:

\medskip
\noindent
Case (i):  $\widetilde{\omega}=0$.
The system admits two rational first integrals,
\begin{equation}
\begin{split}
\label{eq:II1}
F_1 &= y_2, \\
F_2 &= (x_1 y_1 - x_2 y_2)y_2 
      + \frac{k}{2(k+1)}x_1^{k+1} + \frac{m}{2(m+1)}x_1^{m+1}.
\end{split}
\end{equation}

\medskip
\noindent
Case (ii):  $\widetilde{\omega}\neq 0$.  
The system admits two (generally non-algebraic) first integrals,
\begin{equation}
\label{eq:II2}
\begin{split}
I_1&=2y_2\exp\left(\frac{\widetilde{\omega}\, x_1}{2y_2}\right),\\
I_2&=\widetilde{\omega}\left(x_1 y_1-x_2y_2\right)-\left(\frac{I_1}{k\,\widetilde{\omega}}\right)^k\mathrm{\Gamma}\left(k+1,\frac{k\,\widetilde{\omega}\, x_1}{2y_2}\right)\\
&
-\left(\frac{I_1}{m\,\widetilde{\omega}}\right)^m\mathrm{\Gamma}\left(m+1,\frac{m\,\widetilde{\omega}\, x_1}{2y_2}\right),
\end{split}
\end{equation}
where $\mathrm{\Gamma}(s,u)$ is the incomplete gamma function.\end{proposition}

\begin{proof}
The equations of motion generated by Hamiltonian~\eqref{eq:ham_ex}, have the form
\begin{equation}
\label{eq:vhn}
\begin{split}
\dot x_1 &=2y_2-\widetilde{\omega}\, x_1, &&\dot y_1=\widetilde{\omega}\, y_1-k\, x_1^{k-1}-m\, x_1^{m-1},\\
\dot x_2 &=2y_1+\widetilde{\omega}\, x_2. && \dot y_2=-\widetilde{\omega}\, y_2.
\end{split}
\end{equation}
We now divide the proof into two separate cases.
\paragraph{\textit{Case (i)}:}   $\widetilde{\omega}=0$.  
In this case, it is straightforward to verify that $x_2$ is a cyclic variable and the corresponding momentum $y_2$ is a constant of motion.  

To reconstruct the integral $F_2$ given in~\eqref{eq:II1}, we first compute
\begin{align*}
\frac{\rmd}{\rmd t}(x_1y_1 - x_2y_2)
&= x_1\dot y_1 + \dot x_1y_1 - x_2\dot y_2 - \dot x_2y_2 \\
&= -k\,x_1^{k} - m\,x_1^{m}.
\end{align*}
Multiplying this expression by the constant $y_2$ yields
\[
\frac{\rmd}{\rmd t}\!\bigl[(x_1y_1 - x_2y_2)\,y_2\bigr]
= -y_2\bigl(k\,x_1^{k} + m\,x_1^{m}\bigr).
\]
On the other hand, using $\dot x_1 = 2y_2$, we obtain
\begin{align*}
&\frac{\rmd}{\rmd t}\!\left(\frac{k}{2(k+1)}x_1^{k+1}\right)=
 \frac{k}{2}\,x_1^{k}\,\dot x_1
 = k\,x_1^{k}\,y_2,\\
&\frac{\rmd}{\rmd t}\!\left(\frac{m}{2(m+1)}x_1^{m+1}\right)
= \frac{m}{2}\,x_1^{m}\,\dot x_1
 = m\,x_1^{m}\,y_2.
\end{align*}
Adding these three derivatives together, we find
\begin{equation*}
\begin{split}
& \frac{\rmd}{\rmd t}\!\bigl[
(x_1y_1 - x_2y_2)\,y_2  
+ \frac{k}{2(k+1)}x_1^{k+1} 
+ \frac{m}{2(m+1)}x_1^{m+1}\, 
\!\bigr]=0,
\end{split}
\end{equation*}
which confirms that $F_2$ defined in~\eqref{eq:II1}
is the first integral of the system for $\widetilde{\omega}=0$. 
\paragraph{\textit{Case (ii)}:}  $\widetilde{\omega}\neq 0$. 
From~\eqref{eq:vhn} we have
\[
\dot x_1 = 2y_2-\widetilde{\omega}\,x_1,\qquad 
\dot y_2 = -\widetilde{\omega}\,y_2 .
\]
Introduce the auxiliary variable
\begin{equation}
\label{eq:aux}
z:=\frac{\widetilde{\omega}\,x_1}{2y_2}.
\end{equation}
A direct calculation gives
\[
\dot z=\frac{\widetilde{\omega}}{2}\!\left(\frac{\dot x_1}{y_2}-\frac{x_1\dot y_2}{y_2^{\,2}}\right)
=\frac{\widetilde{\omega}}{2}\!\left(\frac{2y_2-\widetilde{\omega}x_1}{y_2}+\frac{\widetilde{\omega}x_1}{y_2}\right)
=\widetilde{\omega}.
\]Consider
\begin{equation}
\label{eq:weds}
I_1 : = 2y_2\,\exp\!\left(\frac{\widetilde{\omega}\,x_1}{2y_2}\right)=2y_2\,\mathrm{e}^z.
\end{equation}
Then
\[
\frac{\mathrm{d}}{\mathrm{d}t}\ln I_1
= \frac{\dot y_2}{y_2} + \dot z
= -\widetilde{\omega} + \widetilde{\omega}
= 0,
\]
which shows that $I_1$ is a first integral, i.e.\ $\dot I_1 = 0$.

Next, we define
\[
J := \widetilde{\omega}(x_1y_1 - x_2y_2).
\]
Using the full system~\eqref{eq:vhn}, one finds
\[
\frac{\mathrm{d}}{\mathrm{d}t}J
= \widetilde{\omega}\,\frac{\mathrm{d}}{\mathrm{d}t}(x_1y_1 - x_2y_2)
= -\,\widetilde{\omega}\bigl(k\,x_1^{k} + m\,x_1^{m}\bigr).
\]
To compensate this derivative, we make use of the derivative of the upper incomplete gamma function.  
Recall that
\[
\frac{\mathrm{d}}{\mathrm{d}u}\mathrm{\Gamma}(s,u)
= -u^{\,s-1}\mathrm{e}^{-u},
\]
see  e.g.~\cite[p.~339]{Magnus_1966}.
Introducing the auxiliary variable~\eqref{eq:aux}, for which $\dot z = \widetilde{\omega}$, we define
\begin{equation*}
\label{eq:dd}
\begin{split}
T_k &:= \left(\frac{I_1}{k\,\widetilde{\omega}}\right)^{\!k}
       \mathrm{\Gamma}\!\left(k+1,\,kz\right),\\[0.5em]
T_m &:= \left(\frac{I_1}{m\,\widetilde{\omega}}\right)^{\!m}
       \mathrm{\Gamma}\!\left(m+1,\,mz\right).
\end{split}
\end{equation*}
Since $I_1$ is constant, by the chain rule we have
\begin{align*}
\dot T_k
&= \left(\frac{I_1}{k\,\widetilde{\omega}}\right)^{\!k}
   \frac{\mathrm{d}}{\mathrm{d}t}\mathrm{\Gamma}(k+1, kz)\\
&= \left(\frac{I_1}{k\,\widetilde{\omega}}\right)^{\!k}
   \!\left[-(kz)^{k}\mathrm{e}^{-kz}\right] k\,\dot z\\
&= -k\,\widetilde{\omega}
   \left[(kz)^{k}\mathrm{e}^{-kz}
   \left(\frac{I_1}{k\,\widetilde{\omega}}\right)^{\!k}\right]
= -k\,\widetilde{\omega}\,x_1^{k},
\end{align*}
where relations~\eqref{eq:aux} and~\eqref{eq:weds} were used.  
Analogously,
\[
\dot T_m = -m\,\widetilde{\omega}\,x_1^{m}.
\]
Combining these with~$\dot J$, we get
\[
\frac{\mathrm{d}}{\mathrm{d}t}\bigl(J - T_k - T_m\bigr)
= -\widetilde{\omega}(k\,x_1^{k} + m\,x_1^{m})
+ \widetilde{\omega}(k\,x_1^{k} + m\,x_1^{m})
= 0.
\]
Hence, a  non-trivial function
\begin{equation}
\label{eq:cal}
\begin{split}
I_2 \coloneqq & \, J - T_k - T_m\\ =\,  & \widetilde{\omega}(x_1y_1-x_2y_2)
-\left(\frac{I_1}{k\,\widetilde{\omega}}\right)^{\!k}\mathrm{\Gamma}\!\left(k+1,\frac{k\,\widetilde{\omega}\,x_1}{2y_2}\right)
\\ - &  \left(\frac{I_1}{m\,\widetilde{\omega}}\right)^{\!m}\mathrm{\Gamma}\!\left(m+1,\frac{m\,\widetilde{\omega}\,x_1}{2y_2}\right),
\end{split}
\end{equation}
is a first integral of the system~\eqref{eq:ham_ex} for $\widetilde{\omega}\neq 0$.

\end{proof}
\begin{remark}
For integer values of the parameters $k,m\in\mathbb{N}$, 
the incomplete gamma functions appearing in~\eqref{eq:II2} 
reduce to finite polynomials.  
Indeed, for any integer~$n\geq~0$ one has the known identity
\[
e^z\, \mathrm{\Gamma}(n+1,z) = n!\,\sum_{j=0}^{n}\frac{z^{j}}{j!},
\]
see e.g.~\cite{Magnus_1966}.
Substituting this expression into~\eqref{eq:II2} and using formula for $I_1$~\eqref{eq:weds}, 
each term of the form
\[
\left(\frac{I_1}{n\,\widetilde{\omega}}\right)^{\!n}
\mathrm{\Gamma}\!\left(n+1,\frac{n\,\widetilde{\omega}\,x_1}{2y_2}\right),
\]can be written as
\begin{align*}
\label{eq:poly-exp-final}
\left(\frac{I_1}{n\,\widetilde{\omega}}\right)^{\!n}
\mathrm{\Gamma}\!\left(n+1,\frac{n\,\widetilde{\omega}\,x_1}{2y_2}\right)
&= n!\sum_{j=0}^{n}\frac{1}{j!}
   \left(\frac{2y_2}{n\,\widetilde{\omega}}\right)^{\!n-j}
   x_1^{\,j}.
\end{align*}
Hence,  for integers $k$ and $m$, the first integral $I_2$~\eqref{eq:II2} is a polynomial function of the form
\begin{align}
\begin{split}
 &I_2= \widetilde{\omega}(x_1y_1-x_2y_2) \\[0.2em]
& 
- k!\sum_{j=0}^{k} \frac{1}{j!}
   \left(\frac{2y_2}{k\,\widetilde{\omega}}\right)^{\!k-j} x_1^{\,j}
\;-\;
  m!\sum_{j=0}^{m} \frac{1}{j!}
   \left(\frac{2y_2}{m\,\widetilde{\omega}}\right)^{\!m-j} x_1^{\,j}.
\end{split}
\end{align}
For non-integers $k,m$, the integral $I_2$ is transcendental. 
\end{remark}

\begin{remark}
Hamiltonian~\eqref{eq:ham_ex} with the exceptional potential~\eqref{eq:exceptional}
does not have a customary form encountered in classical Euclidean mechanics.
In particular, potentials with complex coefficients do not admit a direct physical
realization within the standard framework of Newtonian or Hamiltonian mechanics.

The super-integrable case identified in this work, defined by the Hamiltonian~(79),
should therefore be understood primarily as an analytically consistent example
within the complexified Hamiltonian formalism.
After the appropriate complex transformation, the potential \(V\) becomes real and
takes the form of an anisotropic oscillator with exponents \(m\) and \(k\).
However, under this transformation, the gyroscopic term acquires an imaginary
coefficient.
Nevertheless, in the special case \(\omega = 0\), the transformed Hamiltonian
admits a clear and physically meaningful interpretation.

To provide an alternative geometric interpretation of Hamiltonians with exceptional
potentials, let us consider a point mass moving in a plane equipped with a
Lorentzian metric of signature $(+,-)$, that is, in the Minkowski plane.
In this setting, the Lagrangian function is given by
\begin{equation*}
L=\frac{1}{2}(\dot{q}_1^2-\dot{q}_2^2)-U(q_1,q_2),
\end{equation*}
where $U(q_1,q_2)$ denotes the potential energy.
Introducing new coordinates
\begin{equation*}
q_1 = \frac{1}{2}(x_1+x_2),\quad
q_2 = \frac{1}{2}(x_1 - x_2),
\end{equation*}
we obtain
\begin{equation*}
L=2\dot{x}_1\dot{x}_2 -V(x_1,x_2),
\end{equation*}
with
\begin{equation*}
V(x_1,x_2)=U\left(\frac{x_1+x_2}{2},\frac{x_1 - x_2}{2}\right).
\end{equation*}
The corresponding Hamiltonian takes the form
\begin{equation*}
H = 2p_1 p_2 + V(x_1,x_2).
\end{equation*}
We emphasize that exceptional potentials of this type have been extensively studied
in the literature on integrability and super-integrability, see e.g.~\cite{
AcostaHumanez_2021,Llibre_Valls_2015,Nakagawa_Maciejewski_Przybylska_2005}
.
\end{remark}
 \section{Conclusions}
\label{sec:conclusions}
In this paper, we investigate the integrability of a two-dimensional Hamiltonian system that combines a gyroscopic term and a Keplerian part with a non-homogeneous potential composed of two homogeneous components of different degrees. 
By employing a combination of analytical tools — including the Levi–Civita regularisation, the coupling through constant metamorphosis, and the differential Galois theory — we established explicit obstructions to Liouville integrability.  
These obstructions are expressed in terms of the degrees of homogeneity and the coefficients of the potential, providing a compact criterion that can be directly applied to a wide class of  Hamiltonian systems with a gyroscopic coupling.

The obtained results show that the addition of a gyroscopic term, together with the Kepler-type potential, generally destroys the integrable structure of classical non-rotating systems such as the H\'enon–Heiles and Armbruster–Guckenheimer–Kim models.  
Only a few exceptional parameter configurations remain compatible with the necessary integrability conditions.  
The analytical predictions were further supported by numerical studies based on Poincar\'e cross-sections, which clearly illustrate the breakdown of invariant tori and the emergence of chaotic regions as the strength of the gyroscopic and non-homogeneous terms increases.  
Interestingly, in the absence of the Kepler-type potential, a particular non-homogeneous extension of the exceptional potential remains integrable.  
For this model, we obtained explicit analytic expressions for the first integrals, which are generally transcendental, but we show that for integer homogeneity degrees, they reduce to purely polynomial forms.  

The present work thus provides a unified and systematic framework for studying the (non-)integrability of rotating and non-homogeneous Hamiltonian systems.  
It extends several known results from classical galactic and astrophysical dynamics and offers a coherent mathematical explanation for the loss of integrability in the presence of rotation and anisotropy.  
At the same time, it identifies a number of specific parameter domains where integrability may still persist, opening the way to a more detailed classification of exceptional cases.

As a natural continuation of this research, an important open problem is the study of similar models in the relativistic regime.  
Relativistic corrections are known to affect the integrability and stability properties of dynamical systems, as demonstrated in recent comparative studies between classical and relativistic particle dynamics in flat and curved spaces~\cite{Szuminski_2024_NonlinearDyn,Przybylska_2023_DestructiveRelativity,Szuminski_2025_NonlinearDyn}.   
Since the gyroscopic term has a clear interpretation both in electrodynamics and in astrophysical models, its inclusion in a relativistic Hamiltonian framework may reveal new effects on the structure of first integrals and the onset of chaos.  
In particular, a relativistic version of the present model, possibly formulated on spaces of constant curvature, either within a special-relativistic Hamiltonian framework or on fixed curved configuration manifolds, would be especially interesting from both the physical and mathematical viewpoints.  
Such an extension could bridge the gap between the study of integrability in classical rotating systems and modern approaches to relativistic galactic dynamics, where the curvature of space–time and rotational effects play a crucial role.

In conclusion, the results presented here highlight how the combined influence of rotation, non-homogeneity, and central forces governs the transition between order and chaos in two-dimensional Hamiltonian systems.  
They also point towards new directions — especially the relativistic and curved-space generalisations — in which the interplay between gyroscopic effects and geometry may lead to qualitatively new integrability phenomena.

\section*{Funding} 
This research was funded by the National Science Center of Poland under Grant No. 2020/39/D /ST1/01632, and  by the Minister of Science under
the ‘Regional Excellence Initiative’ program, Project No.
RID/SP/0050/2024/1.
For Open Access, the author has applied a CC-BY public copyright license to any Author Accepted Manuscript (AAM) version arising from this submission.

\section*{Author Contributions }
W. Sz. was responsible for the conceptualisation of the study, the integrability analysis, software development, numerical simulations, validation of the results, and funding acquisition.
A.~J.~M. developed the formal analysis and methodology, prepared the initial manuscript draft, and contributed to funding acquisition.
All authors reviewed and approved the final manuscript.

	\section*{Data availability }
	The data that support the findings of this study are available from the corresponding author,  upon reasonable request.
\section*{Declarations}
	\textbf{Conflicts of interest} The authors declare that they have no conflict of interest.\\

\appendix
\section{Criterion}
\label{app:criterion}
\renewcommand\theequation{A\arabic{equation}}
\setcounter{equation}{0}
Let us consider the following system of differential equations
\begin{subequations}
\label{eq:reducedA}
\begin{align}
\label{eq:reduced_A}
X''&=r(z) X,\\ \label{eq:reduced_B}
Y''&=r(z) Y+s(z)X.
\end{align}
\end{subequations}
where $r(z)$ is a rational function and $s(z)$ is an algebraic function. We
assume that equation~\eqref{eq:reduced_A} is reducible and one of its solutions
is algebraic and the second is transcendent. In the language of differential
Galois theory: the identity component of the differential Galois group of the
system~\eqref{eq:reducedA} is the additive subgroup of $\mathrm{SL}(2,\C)$. The
question is whether the identity component of the differential Galois group of
the system~\eqref{eq:reducedA} is Abelian. To answer this question, we have to
analyze all solutions of the system $(x_1,x_2, y_1,y_2)$, where $x_1$ and $x_2$
are solutions of \eqref{eq:reduced_A}, and  $y_1$ and $y_2$ are solutions of
\eqref{eq:reduced_B}. Using the variation of constants method, we find that 
\begin{equation}
  \label{eq:23A} 
  y_i =
  x_i \int \frac{\varphi_i}{x_i^2} ,
  \qquad \varphi_i = \int s(z)x_i^2, \qquad i=1,2.
\end{equation}
A similar question, in more general settings was investigated in
\cite{Duval:09::}. Results of this paper were used in \cite{2023NonDy.111..275P}to derive the necessary and sufficient conditions which guarantee that the identity component of the differential Galois group of
the system~\eqref{eq:reducedA} is Abelian.
For their  formulation we have to introduce the following  functions. Let $x_1(z)$ be an algebraic solution of equation~\eqref{eq:reduced_A}, and 
\begin{equation}
  \label{eq:psiphi}
	\begin{split}
		\psi(z)=&\int \frac{\rmd z}{x_1(z)^2}, \qquad 
	\varphi(z)=\int s(z){x_1(z)^2}\rmd z, \\
	I(z) =&  \int \psi'(z)\varphi(z) \,\rmd z, \qquad x_2(z)=x_1(z)\psi(z).
	\end{split}
\end{equation}
With the above notation, we have the following.
\begin{lemma}
  \label{lem:necsu0}
  If integral $\psi(z)$ is algebraic, then the identity component of the
  differential Galois group of system~\eqref{eq:reducedA} is Abelian.
\end{lemma}
This lemma follows from Theorem~3.1 of \cite{Duval:09::}.
\begin{lemma}
	\label{lem:necsu1}
	Assume that equation~\eqref{eq:reduced_A} has one algebraic solution $x_1(z)$
	and one transcendent solution $x_2(z)$. If the identity component of the
	differential Galois group of system~\eqref{eq:reducedA} is Abelian then the
	function 
\begin{equation}
	\label{eq:gzA1} 
	g(z) = \varphi(z)+\lambda\psi(z),  
\end{equation}
is algebraic for a certain $\lambda\in\C$. 
\end{lemma}
This is only a necessary condition. If $\varphi(z)$ is algebraic then we can
take $\lambda=0$ and it is fulfilled. 

If $\varphi(z)$ is not algebraic, then we have to use stronger condition.
\begin{lemma}
	\label{lem:necsu2}
	Assume that equation~\eqref{eq:reduced_A} has one algebraic solution $x_1(z)
	$, the second one $x_2(z)$ is transcendent, and $\varphi(z)$ is
	algebraic. Then the identity component of the differential Galois group of system~\eqref{eq:reducedA} is Abelian if and only if function
\begin{equation}
	\label{eq:gzA2} 
	g(z) = \lambda\psi(z) +  I(z),  
\end{equation}
is algebraic for a certain $\lambda\in\C$. 
\end{lemma}
Notice, that if $I(z)$ is algebraic, then we can take $\lambda =0$ and get algebraic $g(z)$. 
\section{Monodromy group of the hypergeometric equation}
\label{app:monodromy}

The Gauss hypergeometric function $F(\alpha,\beta;\gamma;z)=$ 
  ${}_2F_1(\alpha,\beta;\gamma;z) $ is defined by the following
series
\begin{equation}
  \label{eq:2}
  \begin{split}
    F(\alpha,\beta;\gamma &;z)=
    \sum_{n=0}^{\infty}\frac{{\left(\alpha\right)_{n}}{\left(\beta
          \right)_{n}}}{{\left(\gamma\right)_{n}}n!}z^n\\
    =&1+\frac{\alpha\beta}{\gamma}z+\frac{\alpha(\alpha+1)\beta(\beta+1)}
    {\gamma(\gamma+1)2!}z^{2}+\cdots\\ =&
    \frac{\mathrm{\Gamma}\left(\gamma\right)}{\mathrm{\Gamma}\left(\alpha\right)\mathrm{\Gamma}
      \left(\beta\right)}
    \sum_{n=0}^{\infty}\frac{\mathrm{\Gamma}\left(\alpha+n\right)\mathrm{\Gamma}\left(\beta+n
      \right)}{\mathrm{\Gamma}\left(\gamma+n\right)n!}z^{n},
  \end{split}
\end{equation}
where $ (x)_n $
and $ \mathrm{\Gamma}(x) $ denote the Pochhammer symbol and the Euler  gamma function,
respectively. It is a solution of the Gauss hypergeometric
equation%
\begin{equation}
  \label{eq:33}
  z(1-z)\frac{{\mathrm{d}}^{2}w}{{\mathrm{d}z}^{2}}+
  \left(\gamma-(\alpha+\beta+1)z\right)
  \frac{\mathrm{d}w}{\mathrm{d}z}-\alpha\beta w=0,
\end{equation}
which is holomorphic in
the disk $ \abs{z}<1 $. This equation has three regular singularities at $ z
  \in\{ 0,1,\infty\} $, with the corresponding exponent pairs
\begin{equation}
  \label{eq:32} (0, 1 - \gamma), \qquad ( 0, \gamma - \alpha - \beta), \qquad
  (\alpha, \beta ).
\end{equation}
We assume none of $\gamma$,
$\gamma-\alpha-\beta$, $\alpha-\beta$ is an integer.

The effect of analytical continuation of function $ f(z) $ holomorphic in a
domain $U$ along a loop $ \sigma $ with a base point in $ U $ is a function $
\widetilde{f}(z)$. We will denote $\widetilde{f}(z)=\cM_{\sigma}(f(z))$, and $
\cM_{\sigma} $ will be called the monodromy operator. If $ f_1(z) $ and $ f_2(z)
$ span two-dimensional vector space and $ \cM_{\sigma}(f_i(z))=f_1(z)m_{1i} +
f_2(z)m_{2i} $ for $ i=1,2 $. In matrix notation we write 
\begin{equation}
  \label{eq:monodromy_matrix}
  \cM_{\sigma}(\vf(z))= \vf(z)M_{\sigma}, \quad\text{where}\quad
  \vf(z)=[f_1(z),f_2(z)],
\end{equation}
where 
\begin{equation}
\vf(z)=[f_1(z),f_2(z)], \quad\text{and}\quad
  M_{\sigma}=
  \begin{bmatrix}
    m_{11} & m_{12}\\ m_{21} & m_{22}
  \end{bmatrix}.
\end{equation}
Matrix $M_{\sigma}  \in\mathrm{GL}(2,\C) $ is called the monodromy matrix. The is called the monodromy matrix. 
\begin{figure}
  \centering
  \includegraphics[width=0.4\textwidth]{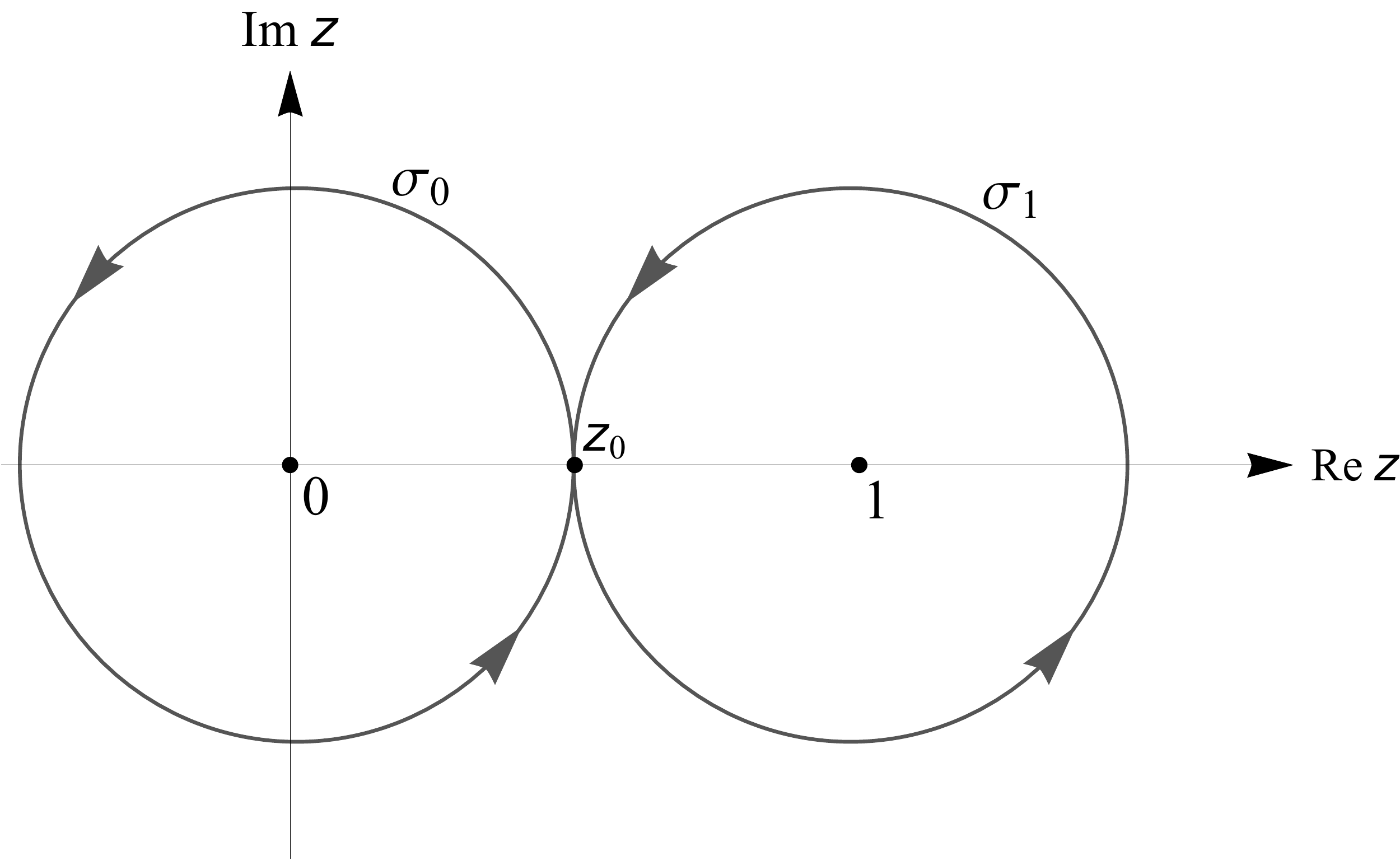}
  \caption{Loops $\sigma_0$ and $\sigma_1$, around singularities $z=0$ and $z=1$.}
  \label{fig:mono}
\end{figure}

We consider two loops $ \sigma_0 $ and $\sigma_1$ with one common point
$z_0\in(0,1)$ that encircle singular points $z=0$ and $z=1$ counter-clockwise,
respectively,  see Figure~\ref{fig:mono}. The third loop $\sigma_{\infty}$
encircles clockwise both singularities $z=0$ and $z=1$. Then in a neighborhood
of each singularity, one can select two independent solutions  of the
hypergeometric equation. They form the fundamental matrices, and their
analytical continuations determine the local monodromy matrices. Appropriate
expressions for all singularities are as follows. 
\begin{itemize}
  \item
        Singularity $z=0$
        \begin{equation}
          \label{eq:sol_zero}
          \begin{split}
            u_1(z)
            :=&F\left(\alpha,\beta;\gamma;z\right), \\ u_2(z):=
            &{z^{1-\gamma}}F\left({\alpha-\gamma+1,\beta-\gamma+1; 2-\gamma};z\right).
          \end{split}
        \end{equation}
        The local
        monodromy matrix at this singularity is%
        \begin{equation}
          \label{eq:M_zero} M_{\sigma_0} =
          \begin{bmatrix}
            1 & 0 \\ 0 &
               c^{-1}
          \end{bmatrix}.
        \end{equation}
  \item
        Singularity $z=1$ %
        \begin{equation}
          \label{eq:sol_one}
          \begin{split}
            v_{1}(z):=&F\left({\alpha,\beta;
              \alpha+\beta-\gamma+1};1-z\right), \\
            v_{2}(z):=&(1-z)^{\gamma-\alpha-\beta}F\left({\gamma-\alpha,\gamma-\beta;
              \gamma-\alpha-\beta+1};1-z\right).
          \end{split}
        \end{equation}
        The local monodromy matrix at
        this singularity is%
        \begin{equation}
          \label{eq:M_one}
          \widetilde{M}_{\sigma_1} =
          \begin{bmatrix}
            1 & 0 \\ 0 & \frac{c}{a b}
          \end{bmatrix}.
        \end{equation}

  \item Singularity $z=\infty$ %
        \begin{equation}
          \label{eq:6}
          \begin{split}
            w_{1}(z):=&{\mathrm{e}}^{\alpha\pi\mathrm{i}}z^{-\alpha}\*F\left({\alpha,\alpha-\gamma+1;
              \alpha-\beta+1}; \tfrac{1}{z}\right),\\%
            w_{2}(z):=&{\mathrm{e}}^{\beta\pi\mathrm{i}}z^{-\beta}F\left({\beta,\beta-\gamma+1;
              \beta-\alpha+1};\tfrac{ 1}{z}\right).%
          \end{split}
        \end{equation}
        The local
        monodromy matrix at this singularity is%
        \begin{equation}
          \label{9} \widetilde{M}_{\sigma_\infty}=
          \begin{bmatrix}
            a & 0
            \\ 0 & b
          \end{bmatrix}.
        \end{equation}
\end{itemize}
In the above formulae we denoted $ a:={\mathrm{e}}^{2\pi\mathrm{i} \alpha} $, $
  b:={\mathrm{e}}^{2\pi\mathrm{i} \beta} $ and $ c:={\mathrm{e}}^{2\pi\mathrm{i}
  \gamma} $.

To determine global monodromy group $\scM$ we fix the basis solutions
$[u_1(z), u_2(z)]$. Then we express all the monodromy matrices with respect to
this basis. Clearly, $M_{\sigma_0}$ is given by~\eqref{eq:M_zero}.  In order to
calculate the monodromy matrix $M_{\sigma_1}$ we need the connection formula
\begin{equation}
  \label{eq:8} 
  \vu(z)=[u_1(z), u_2(z)] = [v_1(z), v_2(z)]\vP = \vv(z)\vP,
\end{equation}
where $\vP$ is the connection matrix
\begin{equation}
  \label{eq:9}
  \vP =\begin{bmatrix}
 \dfrac{\mathrm{\Gamma} (\gamma ) \mathrm{\Gamma} (-\alpha -\beta +\gamma )}{\mathrm{\Gamma} (\gamma
   -\alpha ) \mathrm{\Gamma} (\gamma -\beta )} & \dfrac{\mathrm{\Gamma} (2-\gamma ) \mathrm{\Gamma}
   (-\alpha -\beta +\gamma )}{\mathrm{\Gamma} (1-\alpha ) \mathrm{\Gamma} (1-\beta )} \\[1em]
 \dfrac{\mathrm{\Gamma} (\gamma ) \mathrm{\Gamma} (\alpha +\beta -\gamma )}{\mathrm{\Gamma} (\alpha )
   \mathrm{\Gamma} (\beta )} & \dfrac{\mathrm{\Gamma} (2-\gamma ) \mathrm{\Gamma} (\alpha +\beta
   -\gamma )}{\mathrm{\Gamma} (\alpha -\gamma +1) \mathrm{\Gamma} (\beta -\gamma +1)} 
  \end{bmatrix},
\end{equation}
see \cite[Ch.~2.10]{Erdelyi:81V1::}, or \cite{MR4414300}. Then we have
\begin{equation}
  \label{eq:10}
   \begin{split}
    \cM_{\sigma_1}(\vu(z))&= \cM_{\sigma_1}(\vv(z)\vP)=\cM_{\sigma_1}(\vv(z))\vP\\
    &=
     \vv(z)\widetilde{M}_{\sigma_1}\vP= \vu(z)\vP^{-1} \widetilde{M}_{\sigma_1}\vP(z),
  \end{split}
\end{equation}
and hence
\begin{equation}
  \label{eq:10} M_{\sigma_1}:=
  \vP^{-1}\widetilde{M}_{\sigma_1}\vP.
\end{equation}
The loop around the infinity $\sigma_{\infty}$ is chosen such that
$\sigma_0\sigma_1\sigma_{\infty}=\id$. Then,
the monodromy matrix $M_{\sigma_\infty}$ is given by
\begin{equation}
  \label{eq:11A}
  M_{\sigma_\infty} = M_{\sigma_0}^{-1} M_{\sigma_1}^{-1}
\end{equation}

Now, we consider the special case when parameters $(\alpha,\beta,\gamma)$ are given by
\begin{equation}
\label{eq:case1}
(\alpha,\beta,\gamma) = \left(-\frac{1}{2},1+\frac{l}{2 },\frac{1}{2}\right).
\end{equation} 
In our analysis we follow results given in~\cite{MR4414300}. As 
\begin{equation}
  \label{eq:13}
  \gamma -\alpha -\beta = -\frac{l}{2}, \qquad \alpha-\beta = -\frac{l+3}{2},       
\end{equation}
the above formulae for local solutions and local monodromy matrices are valid
except the case when $l$ is an even integer. 
If $l=2l'\geq 0$ is an integer
$l'$, then local solutions  the local monodromy at $z=0$ are given
by~\eqref{eq:sol_zero} and by~\eqref{eq:M_zero}, respectively. However, in this
case  singularity at $z=1$ is logarithmic and local solutions given by
by~\eqref{eq:sol_one} coincide. Thus, as new basis we take $v_1(z)$ and
$\widetilde{v}_2(z)$ which we determine  using formulae (3.14) and (3.15)
from~\cite{MR4414300}. It has the form
\begin{equation}
  \label{eq:log}
\widetilde{v}_2(z) =\sqrt{z}\ln(1-z) + h(z)
\end{equation}
where $h(z)$ is an  function holomorphic at $z=1$. In this basis the local
monodromy is 
\begin{equation}
  \label{eq:15}
  \widetilde{M}_{\sigma_1} =
  \begin{bmatrix}
    1 & 0 \\ 2\pi \rmi& 1
  \end{bmatrix}.
\end{equation}
The connection matrix can be derive from formula 4.6 from~\cite{MR4414300} and has the form
\begin{equation}
  \label{eq:16}
  \vP =\begin{bmatrix}
    (2-\ln(4)) p_{21}& 1\\
p_{21}& 0
  \end{bmatrix}, \qquad p_{21}= -\dfrac{(-1)^{l'}\sqrt{\pi } }{l'! \mathrm{\Gamma} \left(-l'-\frac{1}{2}\right)}
\end{equation}

For parameters
\begin{equation}
\label{eq:case2}
(\alpha,\beta,\gamma)=(\widehat\alpha,\widehat\beta,\widehat\gamma) = \left(\frac{1}{2},1-\frac{n}{2},\frac{3}{2}\right),
\end{equation}
we have
\begin{equation}
  \label{eq:14}
  \widehat\gamma -\widehat\alpha -\widehat\beta = \frac{n}{2}, \qquad
  \widehat\alpha-\widehat\beta = \frac{n-1}{2}.
\end{equation}
Hence, if $n$ is not an even integer, then local solutions and local monodromy
matrices are given by~\eqref{eq:sol_zero}, \eqref{eq:M_zero},
\eqref{eq:sol_one}, and by~\eqref{eq:M_one}, respectively. Moreover, in this
case also the connection matrix is given by~\eqref{eq:9}. 

If $n=-2n'\leq 0$ is an integer then   local solutions at $z=1$ are
$(v_1(z),\widetilde{v}_2(z))$, where $v_1(z)$ is define by
\eqref{eq:sol_one}where and $\widetilde{v}_2(z)$ has the form~\eqref{eq:log}.
Hence, the local  monodromy matrix at $z=1$ are given by~\eqref{eq:15}. Using
formula 4.6 from~\cite{MR4414300} we can derive the connection matrix. It   has
the form
\begin{equation}
  \label{eq:17}
  \vP =\begin{bmatrix}
    -2 \ln(2) p_{21}& 1\\
p_{21}& 0
  \end{bmatrix}, \qquad p_{21}= -\dfrac{(-1)^{n'}\sqrt{\pi } }{2 n'! \mathrm{\Gamma} \left(\frac{1}{2}-n'\right)}.
\end{equation}


\begin{thebibliography}{10}

\bibitem{SzuminskiMaciejewski2026}
Wojciech Szumiński and Andrzej~J. Maciejewski.
\newblock Integrability of non-homogeneous Hamiltonian systems with gyroscopic coupling.
\newblock {\em Nonlinear Dyn.}, 114:435, 2026.

\bibitem{Morales:00::}
J.~J. Morales-Ruiz.
\newblock Kovalevskaya, Liapounov, Painlev\'e, Ziglin and the differential Galois theory.
\newblock {\em Regul. Chaotic Dyn.} 5(3):251--272, 2000.

\bibitem{Maciejewski:04::}
A.~J. Maciejewski and M.~Przybylska.
\newblock All meromorphically integrable 2D Hamiltonian systems with homogeneous potential of degree~3.
\newblock {\em Phys. Lett. A} 327:461--473, 2004.

\bibitem{Szuminski:15::}
W.~Szumiński, A.~J. Maciejewski, and M.~Przybylska.
\newblock Note on integrability of certain homogeneous Hamiltonian systems.
\newblock {\em Phys. Lett. A} 379(45--46):2970--2976, 2015.

\bibitem{Yoshida:88::}
H.~Yoshida.
\newblock Nonintegrability of the truncated Toda lattice Hamiltonian at any order.
\newblock {\em Commun. Math. Phys.} 116(4):529--538, 1988.

\bibitem{Mondejar:99::}
F.~Mondejar, S.~Ferrer, and A.~Vigueras.
\newblock On the non-integrability of Hamiltonian systems with sum of homogeneous potentials.
\newblock Technical Report, Univ. Politécnica de Cartagena, 1999.

\bibitem{2023NonDy.111..275P}
M.~Przybylska and A.~J. Maciejewski.
\newblock Integrability of Hamiltonian systems with gyroscopic term.
\newblock {\em Nonlinear Dyn.} 111(1):275--287, 2023.

\bibitem{Caranicolas_1989}
N.~D. Caranicolas.
\newblock A mapping for the study of the 1:1 resonance in a galactic type Hamiltonian.
\newblock {\em Celest. Mech. Dyn. Astron.} 47:87--96, 1989.

\bibitem{Caranicolas_1990a}
N.~D. Caranicolas.
\newblock Exact periodic orbits and chaos in polynomial potentials.
\newblock {\em Astrophys. Space Sci.} 167:305--313, 1990.

\bibitem{Kubu_2024}
O.~Kubu, A.~Marchesiello, and L.~{\v{S}}nobl.
\newblock Integrable systems in magnetic fields: the generalized parabolic cylindrical case.
\newblock {\em J. Phys. A} 57:235203, 2024.

\bibitem{Innanen_1985}
K.~A. Innanen.
\newblock The threshold of chaos for H\'enon--Heiles and related potentials.
\newblock {\em Astron. J.} 90:2377--2380, 1985.

\bibitem{Lanchares_2021}
V.~Lanchares, A.~I. Pascual, M.~In{\~a}rrea, and D.~Farrelly.
\newblock Reeb's theorem and periodic orbits for a rotating H\'enon--Heiles potential.
\newblock {\em J. Dyn. Differ. Equ.} 33:445--461, 2021.

\bibitem{Contopoulos_2002}
G.~Contopoulos.
\newblock {\em Order and Chaos in Dynamical Astronomy}.
\newblock Springer, 2002.

\bibitem{Hill_1905}
G.~W. Hill.
\newblock Researches in the lunar theory.
\newblock In {\em Collected Mathematical Works}, vol.~1, pp.~284--335. Carnegie Inst., Washington, 1905.



\bibitem{Vallado_2007}
D.~A. Vallado.
\newblock {\em Fundamentals of Astrodynamics and Applications}.
\newblock Microcosm Press / Springer, 3rd ed., 2007.

\bibitem{Giacaglia_2019}
G.~E.~O. Giacaglia and W.~Q. Lamas.
\newblock {\em Introduction to Artificial Satellites Dynamics}.
\newblock Independently published, 2019.

\bibitem{Ito_1987}
H.~Ito.
\newblock A criterion for non-integrability of Hamiltonian systems with nonhomogeneous potentials.
\newblock {\em Z. Angew. Math. Phys.} 38:459--476, 1987.

\bibitem{Royer_2011}
A.~Royer.
\newblock Why is the magnetic force similar to a Coriolis force?
\newblock {\em arXiv:1109.3624}, 2011.

\bibitem{Brandao_2015}
J.~E. Brand{\~a}o, F.~Moraes, M.~M. Cunha, J.~R.~F. Lima, and C.~Filgueiras.
\newblock Inertial--Hall effect: influence of rotation on Hall conductivity.
\newblock {\em Results Phys.} 5:55--59, 2015.

\bibitem{Combot_2022}
T.~Combot, A.~J. Maciejewski, and M.~Przybylska.
\newblock Integrability of the generalised Hill problem.
\newblock {\em Nonlinear Dyn.} 107:1989--2002, 2022.

\bibitem{Inarrea_2015}
M.~In{\~a}rrea, V.~Lanchares, J.~F. Palaci\'an, A.~I. Pascual, J.~P. Salas, and P.~Yanguas.
\newblock Lyapunov stability for a generalized H\'enon--Heiles system.
\newblock {\em Appl. Math. Comput.} 253:159--171, 2015.

\bibitem{Binney_Tremaine_2008}
J.~Binney and S.~Tremaine.
\newblock {\em Galactic Dynamics}.
\newblock Princeton Univ. Press, 2nd ed., 2008.

\bibitem{Elmandouh2016}
A.~A. Elmandouh.
\newblock On the dynamics of Armbruster--Guckenheimer--Kim galactic potential.
\newblock {\em Astrophys. Space Sci.} 361:182, 2016.

\bibitem{Salas_2022}
J.~P. Salas, V.~Lanchares, M.~In{\~a}rrea, and D.~Farrelly.
\newblock Coriolis coupling in a H\'enon--Heiles system.
\newblock {\em Commun. Nonlinear Sci. Numer. Simul.} 111:106484, 2022.

\bibitem{Lacomba_2012}
E.~A. Lacomba and J.~Llibre.
\newblock Dynamics of a galactic Hamiltonian system.
\newblock {\em J. Math. Phys.} 53:072901, 2012.

\bibitem{Henon:64::}
M.~H\'enon and C.~Heiles.
\newblock The applicability of the third integral of motion.
\newblock {\em Astron. J.} 69:73--79, 1964.

\bibitem{Armbruster_1989}
D.~Armbruster, J.~Guckenheimer, and S.~Kim.
\newblock Chaotic dynamics in systems with square symmetry.
\newblock {\em Phys. Lett. A} 140:416--420, 1989.

\bibitem{Beletsky_2001}
V.~V. Beletsky.
\newblock {\em Essays on the Motion of Celestial Bodies}.
\newblock Birkhäuser, Basel, 2001.

\bibitem{Chauvineau_1990}
B.~Chauvineau and F.~Mignard.
\newblock Generalized Hill's problem: Lagrangian Hill's case.
\newblock {\em Celest. Mech. Dyn. Astron.} 47:123--144, 1990.

\bibitem{MR2647643}
J.~J. Morales-Ruiz and J.~P. Ramis.
\newblock Integrability of dynamical systems through differential Galois theory.
\newblock {\em Contemp. Math.} 509:143--220, 2010.

\bibitem{Morales:99::}
J.~J. Morales-Ruiz.
\newblock {\em Differential Galois Theory and Non-Integrability of Hamiltonian Systems}.
\newblock Birkhäuser, Basel, 1999.

\bibitem{Put:03::}
M.~van~der Put and M.~F. Singer.
\newblock {\em Galois Theory of Linear Differential Equations}.
\newblock Springer, Berlin, 2003.

\bibitem{Audin:08::}
M.~Audin.
\newblock {\em Hamiltonian Systems and Their Integrability}.
\newblock AMS, 2008.

\bibitem{Maciejewski:03::a}
A.~J. Maciejewski and M.~Przybylska.
\newblock Non-integrability of a rigid satellite.
\newblock {\em Celest. Mech. Dyn. Astron.} 87:317--351, 2004.

\bibitem{Boucher:03::}
D.~Boucher and J.~A. Weil.
\newblock Application of Morales–Ramis theory to the planar three-body problem.
\newblock {\em IRMA Lect. Math. Theor. Phys.} 3:163--177, 2003.

\bibitem{Maciejewski:10::}
A.~J. Maciejewski and M.~Przybylska.
\newblock Partial integrability of Hamiltonian systems with homogeneous potential.
\newblock {\em Regul. Chaotic Dyn.} 15:551--563, 2010.

\bibitem{Maciejewski:11::}
A.~J. Maciejewski and M.~Przybylska.
\newblock Non-integrability of the three-body problem.
\newblock {\em Celest. Mech. Dyn. Astron.} 110:17--300, 2011.

\bibitem{Maciejewski_2025}
A.~J. Maciejewski, M.~Przybylska, and T.~Combot.
\newblock Non-integrability of the $n$-body problem.
\newblock {\em J. Eur. Math. Soc.}, 2025.

\bibitem{Acosta_2018}
P.~Acosta-Hum\'anez, M.~Alvarez-Ram\'irez, and T.~J. Stuchi.
\newblock Nonintegrability of the AGK quartic Hamiltonian.
\newblock {\em SIAM J. Appl. Dyn. Syst.} 17:78--96, 2018.

\bibitem{Yagasaki:18::}
K.~Yagasaki.
\newblock Nonintegrability of the unfolding of the fold–Hopf bifurcation.
\newblock {\em Nonlinearity} 31:341--365, 2018.

\bibitem{Huang:18::}
K.~Huang, S.~Shi, and W.~Li.
\newblock Meromorphic and formal first integrals for the Lorenz system.
\newblock {\em J. Nonlinear Math. Phys.} 25:106--121, 2018.

\bibitem{Combot:18::}
T.~Combot.
\newblock Integrability of the one-dimensional Schrödinger equation.
\newblock {\em J. Math. Phys.} 59:022105, 2018.

\bibitem{Mnasri:18::}
C.~Mnasri and A.~A. Elmandouh.
\newblock Plane motion under potential forces in a magnetic field.
\newblock {\em Results Phys.} 9:825--831, 2018.

\bibitem{Shibayama:18::}
M.~Shibayama.
\newblock Non-integrability of the spatial $n$-center problem.
\newblock {\em J. Differ. Equ.} 264:6891--6909, 2018.

\bibitem{Elmandouh:18::}
A.~A. Elmandouh.
\newblock Integrability of 2D Hamiltonians with variable Gaussian curvature.
\newblock {\em Nonlinear Dyn.} 93:933--943, 2018.

\bibitem{Szuminski:18a::}
W.~Szumiński.
\newblock Integrable and superintegrable weight-homogeneous systems.
\newblock {\em Commun. Nonlinear Sci. Numer. Simul.} 67:600--616, 2019.

\bibitem{Maciejewski:18::}
A.~J. Maciejewski and W.~Szumiński.
\newblock Non-integrability of semiclassical Jaynes–Cummings models.
\newblock {\em Appl. Math. Lett.} 82:132--139, 2018.

\bibitem{Maciejewski:17::}
A.~J. Maciejewski, M.~Przybylska, and W.~Szumiński.
\newblock Anisotropic Kepler and two fixed centres.
\newblock {\em Celest. Mech. Dyn. Astron.} 127:163--184, 2017.

\bibitem{Szuminski:16::}
W.~Szumiński and T.~Stachowiak.
\newblock Analysis of a constrained two-body problem.
\newblock {\em Springer Proc. Math. Stat.} 182:361--372, 2016.

\bibitem{Szuminski:24::}
W.~Szumiński.
\newblock A model of variable-length coupled pendulums.
\newblock {\em J. Sound Vib.} in press, 2024.

\bibitem{Szuminski_2025_JSV}
W.~Szumiński and T.~Kapitaniak.
\newblock Dynamics and non-integrability of the variable-length double pendulum.
\newblock {\em J. Sound Vib.} 611:119099, 2025.

\bibitem{Combot:13::}
T.~Combot.
\newblock Integrability conditions at order~2 for potentials of degree~$-1$.
\newblock {\em Nonlinearity} 26:95--120, 2013.

\bibitem{Maciejewski:16::}
A.~J. Maciejewski and M.~Przybylska.
\newblock Integrability of Hamiltonian systems with algebraic potentials.
\newblock {\em Phys. Lett. A} 380:76--82, 2016.

\bibitem{Hietarinta:87::}
J.~Hietarinta.
\newblock Direct methods for the search of a second invariant.
\newblock {\em Phys. Rep.} 147:87--154, 1987.

\bibitem{Post:2010::}
S.~Post.
\newblock Coupling constant metamorphosis and superintegrability.
\newblock {\em AIP Conf. Proc.} 1323:265--274, 2010.

\bibitem{Sergyeyev:12::}
A.~Sergyeyev.
\newblock Coupling constant metamorphosis for finite-dimensional systems.
\newblock {\em Phys. Lett. A} 376:2015--2022, 2012.

\bibitem{Iwasaki:91::}
K.~Iwasaki, H.~Kimura, S.~Shimomura, and M.~Yoshida.
\newblock {\em From Gauss to Painlev\'e}.
\newblock Vieweg, 1991.

\bibitem{Duval:09::}
G.~Duval and A.~J. Maciejewski.
\newblock Jordan obstruction to integrability.
\newblock {\em Ann. Inst. Fourier} 59:2839--2890, 2009.

\bibitem{Steklain_2006}
A.~F. Steklain and P.~S. Letelier.
\newblock Newtonian and pseudo-Newtonian Hill problem.
\newblock {\em Phys. Lett. A} 352:398--403, 2006.

\bibitem{Steklain_2009}
A.~F. Steklain and P.~S. Letelier.
\newblock Stability of orbits around a spinning body.
\newblock {\em Phys. Lett. A} 373:188--194, 2009.

\bibitem{Zotos_2019}
E.~E. Zotos and A.~F. Steklain.
\newblock Motion in the pseudo-Newtonian Hill system.
\newblock {\em Astrophys. Space Sci.} 364:184, 2019.

\bibitem{Kanavos_2002}
S.~S. Kanavos, V.~V. Markellos, E.~A. Perdios, et~al.
\newblock Photogravitational Hill problem.
\newblock {\em Earth Moon Planets} 91:223--241, 2002.

\bibitem{Markellos_2000}
V.~Markellos, A.~Roy, M.~Velgakis, et~al.
\newblock Radiation effects on Hill stability.
\newblock {\em Astrophys. Space Sci.} 271:293--301, 2000.

\bibitem{Heggie_2001}
D.~C. Heggie.
\newblock Escape in Hill’s problem.
\newblock In {\em The Restless Universe}, pp.~109--128, 2001.

\bibitem{Deng_2021}
Y.~Deng, S.~Ibrahim, and E.~E. Zotos.
\newblock Hill-type lunar problem with homogeneous potential.
\newblock {\em Meccanica} 56:2183--2195, 2021.

\bibitem{Combot:18b::}
T.~Combot and C.~Sanabria.
\newblock A symplectic Kovacic algorithm in dimension~4.
\newblock In {\em ISSAC 2018}, pp.~143--150, ACM, 2018.

\bibitem{Ford:73::}
J.~Ford.
\newblock Transition from analytic dynamics to statistical mechanics.
\newblock {\em Adv. Chem. Phys.} 24:155--183, 1973.

\bibitem{Mattheakis:22::}
M.~Mattheakis, D.~Sondak, A.~S. Dogra, and P.~Protopapas.
\newblock Hamiltonian neural networks for ODEs.
\newblock {\em Phys. Rev. E} 105:065305, 2022.

\bibitem{Chang:82::}
Y.~F. Chang, M.~Tabor, and J.~Weiss.
\newblock Analytic structure of the H\'enon--Heiles Hamiltonian.
\newblock {\em J. Math. Phys.} 23:531--538, 1982.

\bibitem{Grammaticos:83::}
B.~Grammaticos, B.~Dorizzi, and A.~Ramani.
\newblock Integrability of Hamiltonians with cubic and quartic potentials.
\newblock {\em J. Math. Phys.} 24:2289--2295, 1983.

\bibitem{Ito:85::}
H.~Ito.
\newblock Non-integrability of H\'enon--Heiles.
\newblock {\em Kodai Math. J.} 8:120--138, 1985.

\bibitem{Li:11::}
W.~Li, S.~Shi, and B.~Liu.
\newblock Non-integrability of a class of Hamiltonian systems.
\newblock {\em J. Math. Phys.} 52:112702, 2011.

\bibitem{Zotos_2020_HHsingular}
E.~E. Zotos, W.~Chen, J.~F. Navarro, and T.~Saeed.
\newblock H\'enon--Heiles potential with singular terms.
\newblock {\em Int. J. Bifurcation Chaos} 30:2050197, 2020.

\bibitem{Navarro_2021_EPJPlus}
J.~F. Navarro.
\newblock Surface of section for H\'enon--Heiles with singular terms.
\newblock {\em Eur. Phys. J. Plus} 136:573, 2021.

\bibitem{Llibre_2019}
J.~Llibre and C.~Valls.
\newblock Global dynamics of the integrable AGK potential.
\newblock {\em Astrophys. Space Sci.} 364:130, 2019.

\bibitem{Kasperczuk_1994}
S.~Kasperczuk.
\newblock Integrability of the Yang--Mills Hamiltonian system.
\newblock {\em Celest. Mech. Dyn. Astron.} 58:387--391, 1994.

\bibitem{JimenezLara_2011_JMP}
L.~Jim\'enez-Lara and J.~Llibre.
\newblock Periodic orbits and nonintegrability of Yang--Mills systems.
\newblock {\em J. Math. Phys.} 52:032901, 2011.

\bibitem{Maciejewski_2005_JMP}
A.~J. Maciejewski and M.~Przybylska.
\newblock Darboux points and integrability of homogeneous polynomial potentials.
\newblock {\em J. Math. Phys.} 46:062901, 2005.


\bibitem{Magnus_1966}
W.~Magnus, F.~Oberhettinger, and R.~P. Soni.
\newblock {\em Formulas and Theorems for the Special Functions of Mathematical Physics}.
\newblock Springer, 1966.

\bibitem{AcostaHumanez_2021}
P.~B.~Acosta-Hum\'anez, M.~\'Alvarez-Ram\'irez, and T.~J.~Stuchi.
\newblock A note on the integrability of exceptional potentials via polynomial
bi-homogeneous potentials.
\newblock {\em Bull. Comput. Appl. Math.} 9(2):59--75, 2021.

\bibitem{Llibre_Valls_2015}
J.~Llibre and C.~Valls.
\newblock Analytic integrability of Hamiltonian systems with exceptional potentials.
\newblock {\em Phys. Lett. A} 379(38):2295--2299, 2015.

\bibitem{Nakagawa_Maciejewski_Przybylska_2005}
K.~Nakagawa, A.~J.~Maciejewski, and M.~Przybylska.
\newblock New integrable Hamiltonian system with first integral quartic in momenta.
\newblock {\em Phys. Lett. A} 343(1--3):171--173, 2005.


\bibitem{Szuminski_2024_NonlinearDyn}
W.~Szumiński, M.~Przybylska, and A.~J. Maciejewski.
\newblock Chaos and integrability of relativistic homogeneous potentials.
\newblock {\em Nonlinear Dyn.} 112:4879--4898, 2024.

\bibitem{Przybylska_2023_DestructiveRelativity}
M.~Przybylska, W.~Szumiński, and A.~J. Maciejewski.
\newblock Destructive relativity.
\newblock {\em Chaos} 33:063156, 2023.

\bibitem{Szuminski_2025_NonlinearDyn}
W.~Szumiński and M.~Przybylska.
\newblock Integrability and chaos of relativistic systems on 2D manifolds.
\newblock {\em Nonlinear Dyn.} 113:30057--30085, 2025.

\bibitem{Erdelyi:81V1::}
A.~Erd\'elyi and W.~Magnus.
\newblock {\em Higher Transcendental Functions, Vol.~I}.
\newblock Krieger, 1981.

\bibitem{MR4414300}
Y.~Haraoka.
\newblock Connection relations for the Gauss hypergeometric equation.
\newblock {\em Kumamoto J. Math.} 35:1--60, 2022.

\end{thebibliography}
\end{document}